\documentclass[reqno,10pt,centertags]{amsart}
\usepackage{amsmath,amsthm,amscd,amssymb,latexsym,upref,mathrsfs}

\usepackage{hyperref}
\newcommand*{\mailto}[1]{\href{mailto:#1}{\nolinkurl{#1}}}
\newcommand{\arxiv}[1]{\href{http://arxiv.org/abs/#1}{arXiv:#1}}
\newcommand{\doi}[1]{\href{http://dx.doi.org/#1}{DOI:#1}}

\newcommand{\bbN}{{\mathbb{N}}}
\newcommand{\bbR}{{\mathbb{R}}}

\newcommand{\bbC}{{\mathbb{C}}}
\newcommand{\bbT}{{\mathbb{T}}}

\newcommand{\calD}{{\mathcal D}}

\newcommand{\calF}{{\mathcal F}}

\newcommand{\calK}{{\mathcal K}}

\newcommand{\Green}{{\mathcal G}}

\newcommand{\dott}{\,\cdot\,}
  
\newcommand{\no}{\nonumber}
\newcommand{\lb}{\label}
\newcommand{\f}{\frac}
\newcommand{\ul}{\underline}
\newcommand{\ol}{\overline}
\newcommand{\ti}{\tilde}
\newcommand{\wti}{\widetilde}

\newcommand{\kap}{w}

\newcommand{\Oh}{O}

\newcommand{\loc}{\text{\rm{loc}}}

\newcommand{\bi}{\bibitem}

\newcommand{\Pinfp}{{P_{\infty_+}}}
\newcommand{\Pinfm}{{P_{\infty_-}}}

\newcommand{\Pinfpm}{{P_{\infty_\pm}}}

\renewcommand{\Re}{\text{\rm Re}}
\renewcommand{\Im}{\text{\rm Im}}

\DeclareMathOperator{\CH}{CH-2}
\DeclareMathOperator{\sCH}{s-CH-2}

\DeclareMathOperator{\ch1}{CH-1} 
\allowdisplaybreaks
\numberwithin{equation}{section}

\newtheorem{theorem}{Theorem}[section]

\newtheorem{lemma}[theorem]{Lemma}
\newtheorem{corollary}[theorem]{Corollary}
\newtheorem{definition}[theorem]{Definition}
\newtheorem{hypothesis}[theorem]{Hypothesis}

\theoremstyle{remark}
\newtheorem{remark}[theorem]{Remark}

\newcommand{\abs}[1]{\lvert#1\rvert}

\begin{document}

\title[Real algebro-geometric solutions of the CH-2 hierarchy]{Real-Valued
Algebro-Geometric Solutions of the\\ Two-Component Camassa--Holm Hierarchy}

\author[J.\ Eckhardt]{Jonathan Eckhardt}
\address{Faculty of Mathematics\\ University of Vienna\\
Oskar-Morgenstern-Platz 1\\ 1090 Wien\\ Austria}
\email{\mailto{jonathan.eckhardt@univie.ac.at}}
\urladdr{\url{http://homepage.univie.ac.at/jonathan.eckhardt/}}

\author[F.\ Gesztesy]{Fritz Gesztesy}
\address{Department of Mathematics,
University of Missouri, Columbia, MO 65211, USA}
\email{\mailto{gesztesyf@missouri.edu}}
\urladdr{\url{https://www.math.missouri.edu/people/gesztesy}}

\author[H.\ Holden]{Helge Holden}
\address{Department of Mathematical Sciences,
Norwegian University of
Science and Technology, NO--7491 Trondheim, Norway}
\email{\mailto{holden@math.ntnu.no}}
\urladdr{\url{http://www.math.ntnu.no/~holden/}}

\author[A.\ Kostenko]{Aleksey Kostenko}
\address{Faculty of Mathematics\\ University of Vienna\\
Oskar-Morgenstern-Platz 1\\ 1090 Wien\\ Austria}
\email{Oleksiy.Kostenko@univie.ac.at; duzer80@gmail.com}
\urladdr{\url{http://www.mat.univie.ac.at/~kostenko/}}

\author[G.\ Teschl]{Gerald Teschl}
\address{Faculty of Mathematics\\ University of Vienna\\
Oskar-Morgenstern-Platz 1\\ 1090 Wien\\ Austria\\ and International
Erwin Schr\"odinger
Institute for Mathematical Physics\\ Boltzmanngasse 9\\ 1090 Wien\\ Austria}
\email{\mailto{Gerald.Teschl@univie.ac.at}}
\urladdr{\url{http://www.mat.univie.ac.at/~gerald/}} 

\thanks{F.G.\ and H.H.\ were supported in part by the Research Council of Norway. Research of 
J.E.\ and A.K.\ were supported by the Austrian Science Fund (FWF) under 
Grants No.\ J3455 and P26060, respectively.}
\date{\today}

\begin{abstract}
We provide a construction of the two-component Camassa--Holm (CH-2) hierarchy 
employing a new zero-curvature formalism and identify and 
describe in detail the isospectral set associated to all real-valued, smooth, and bounded
algebro-geometric solutions of the $n$th equation of the stationary CH-2 hierarchy as 
the real $n$-dimensional  torus $\mathbb{T}^n$. We employ Dubrovin-type equations for 
auxiliary divisors and certain aspects of direct and inverse spectral theory for self-adjoint 
singular Hamiltonian systems. In particular, we employ Weyl--Titchmarsh theory for singular 
(canonical) Hamiltonian systems. 

While we focus primarily on the case of stationary algebro-geometric CH-2 solutions,
we note that the time-dependent case subordinates to the stationary one with respect to isospectral torus questions. 
\end{abstract}

\maketitle


{\scriptsize{\tableofcontents}}

\section{Introduction}\lb{s1}

The principal purpose of this paper is two-fold: first, we provide a construction of the 
two-component Camassa--Holm ($\CH$) hierarchy based on a new zero-curvature 
pair, and second, identify and describe in detail the isospectral set associated to all real-valued, smooth, and bounded algebro-geometric solutions of the $n$th equation of the stationary $\CH$ hierarchy as the real $n$-dimensional torus $\bbT^n$.   

The first nonlinear partial differential equation of the two-component Camassa--Holm hierarchy, 
the two-component Camassa--Holm system \cite{OR96}, can be written in the form
\begin{align} 
\begin{split} 
& 4 u_t-u_{xxt} - 2 u u_{xxx} - 4 u_x u_{xx} + 24 u u_x + \kap_x = 0,   \lb{ch1.0} \\
& \kap_t + 4 \kap u_x + 2 \kap_x u = 0,   \quad (x,t) \in \bbR^2.
\end{split} 
\end{align}
When studying weak solutions of the Cauchy problem one writes the second equation in conservative form, that is, $\rho_t+2(\rho u)_x=0$ where $w=\rho^2$.  For smooth solutions like those studied in the present paper, the two formulations are equivalent. 
This two-component system extends the Camassa--Holm equation, also known 
as the dispersive shallow water equation \cite{CH93} (the special case $\kap \equiv 0$ of 
\eqref{ch1.0}) given by 
\begin{equation}
4u_t-u_{xxt}-2u u_{xxx}-4u_x u_{xx}+24u u_x=0, \quad (x,t)\in\bbR^2
\lb{ch1.1}
\end{equation}
(choosing a convenient scaling of $x,t$). The two-component $\CH$ system \eqref{ch1.0} has 
generated much interest over the past decades. For instance, its relevance to shallow water theory is discussed in \cite{CI08}, \cite{Io13},  
well-posedness and blow-up 
are studied in \cite{CI08}, \cite{ELY07}, \cite{FQ09}, \cite{GY10}, \cite{Ko14}, various types of solutions (global, dissipative, conservative, etc.) are treated in \cite{ELY07}, \cite{GHR12}--\cite{GHR15}, \cite{GY11}, the inverse scattering transform is 
applied to the $\CH$ system in \cite{Ec15}, \cite{HI11},  
$N$ solitary waves are discussed in \cite{CI08}, 
\cite{HI10}, \cite{HI11}, \cite{Li15}, traveling waves are studied in \cite{Mo09}, \cite{Mu09}, 
the geometry of $\CH$ is investigated in \cite{EKL11}, \cite{HI10}, \cite{HI11}, \cite{Ko15}, the periodic $\CH$ system is discussed in  \cite{GHR13}, \cite{HY11}, \cite{To13}.  
For connections to other integrable systems see \cite{AGZ06}, \cite{CLZ06},  \cite{Fa06}. Various multicomponent extensions of the Camassa--Holm equation and its generalizations  are discussed in, e.g.,  \cite{CW12}, \cite{CLY15}, \cite{FQ09}, \cite{Iv06}, \cite{Ko12}, \cite{MAHZ15}, \cite{Mo15}, \cite{We14}--\cite{YQY15}. Closest to the  investigations in this paper is the derivation of the $\CH$ hierarchy and its algebro-geometric 
solutions in \cite{HF15}. 

In Section \ref{s2} we recall the basic polynomial recursion formalism that  defines the $\CH$ 
hierarchy using a new zero-curvature approach based on the $2\times 2$ matrix pair 
$(U, V_n)$, $n\in\bbN_0$ (with $\bbN_0=\bbN\cup\{0\}$), given by
\begin{align}
\begin{split} 
U(z,x,t) = - z^{-1} \begin{pmatrix} \alpha(x,t) & -1\\
\alpha(x,t)^2 + \kap(x,t) & - \alpha(x,t) \end{pmatrix} + \begin{pmatrix}-1 &0\\
0 &1 \end{pmatrix},&     \\ 
z\in\bbC\backslash\{0\}, \; (x,t) \in \bbR^2,& \lb{ch1.3} 
\end{split} 
\end{align}
where
\begin{equation}
\alpha(x,t)= u_x(x,t) + 2u(x,t), \quad (x,t)\in\bbR^2, 
\end{equation}
 and  
\begin{equation}
V_{n}(z,x,t) = z^{-1} 
\begin{pmatrix}-G_{n+1}(z,x,t)& F_{n}(z,x,t)\\
 H_{n}(z,x,t) & G_{n+1}(z,x,t)
\end{pmatrix}, \quad z\in\bbC\backslash\{0\}, 
\; (x,t) \in \bbR^2, \lb{ch1.4} 
\end{equation} 
assuming $F_n$, $H_n$, and $G_{n+1}$ to be polynomials of degree $n$ 
and $n+1$, respectively, with respect to (the spectral parameter) $z$ and $C^\infty$ in $x,t$ (for simplicity). In addition, $F_n$ and 
$G_{n+1}$ are chosen to be monic with respect to $z \in \bbC$. 
The zero-curvature condition
\begin{equation}
U_t(z,x,t) - V_{n,x,t}(z,x,t) + [U(z,x,t),V_n(z,x,t)]=0, \lb{ch1.5}
\end{equation}
is then shown to generate the $\CH$ hierarchy associated to the system \eqref{ch1.0}. 
In fact, \eqref{ch1.0} corresponds to the first nonlinear system $n=1$ (the case $n=0$ represents 
a linear system).  Actually, we derive the corresponding stationary (i.e., $t$-independent)  
hierarchy first as the latter will be most instrumental in determining the isospectral torus of all real-valued, smooth, and bounded algebro-geometric solutions of the $\CH$ hierarchy. The stationary 
hierarchy is derived from the corresponding zero-curvature equation 
\begin{equation}
- V_{n,x}(z,x) + [U(z,x),V_n(z,x)]=0, \lb{ch1.5a}
\end{equation}
and it in turn naturally leads to the identity,
\begin{equation}
G_{n+1}(z,x)^2 + F_n(z,x) H_n(z,x)=R_{2n+2}(z), \lb{ch1.5b}
\end{equation}
where $R_{2n+2}$ is an $x$-independent monic polynomial with respect to $z$ of degree 
$2n + 2$. The polynomial $R_{2n+2}$ is fundamental as it defines the hyperelliptic curve 
$\calK_n$ (cf.\ \eqref{ch3.3}) underlying the stationary $\CH$ hierarchy.

Section \ref{s3} is devoted to the stationary $\CH$ hierarchy and the associated 
algebro-geometric formalism. In particular, the underlying hyperelliptic curve $\calK_n$ 
(defined in terms of the polynomial $R_{2n + 2}$), 
an associated fundamental meromorphic function $\phi$ on $\calK_n$, its divisor of zeros and poles, 
the Baker--Akhiezer vector $\Psi$, basic properties of $\phi$ and $\Psi$, Dubrovin-type equations for auxiliary Dirichlet divisors (in fact, zeros $\hat \mu_j \in \calK_n$, $j=1,\dots,n$, of $\phi$), trace formulas for $u$ and $\kap$ in terms of the projections $\mu_j \in \bbC$, $j=1,\dots,n$, and 
asymptotic properties of 
$\phi$ and $\Psi$ are derived in detail. We conclude this section with a proof of the fact that 
solutions of the Dubrovin equations generate stationary (algebro-geometric) solutions of the 
stationary $\CH$ hierarchy via the trace formulas \eqref{ch3.36}, \eqref{ch3.36a} for the pair 
$(u,\kap)$.

Section \ref{s4} provides a brief summary of self-adjoint singular canonical systems as needed 
in the subsequent Section \ref{s5}, and introduces (scalar-valued) half-line Weyl--Titchmarsh 
functions as well as their $2 \times 2$ matrix-valued generalizations for the entire real line. 

Finally, Section \ref{s5} contains the principal result of this paper, the identification and description 
of the isospectral set of all real-valued, smooth, and bounded algebro-geometric solutions of the 
$n$th equation of the stationary $\CH$ hierarchy as the real $n$-dimensional torus $\bbT^n$. 
We start this section by noticing that the basic stationary equation \eqref{ch3.22},
\begin{equation}
\Psi_x(-z,x)=U(-z,x)\Psi(-z,x), \quad \Psi=(\psi_1, \psi_2)^\top, \;
(z,x)\in\bbC\times\bbR, \lb{ch1.6} 
\end{equation}
is equivalent to the following singular Hamiltonian (canonical) system 
\begin{equation} 
J \wti\Psi_x(\ti z,x)=[\ti z A(x)+B(x)]\wti \Psi(\ti z,x), \quad 
\wti\Psi= \big(\wti\psi_1,\wti\psi_2\big)^\top, \;
\big(\ti z = - z^{-1},x\big)\in\bbC\times\bbR, \lb{ch1.7} 
\end{equation}
where
\begin{align}
& J=\begin{pmatrix} 0 & -1 \\ 1 & 0 \end{pmatrix}, \quad 
\wti\Psi(\ti z,x)=\Psi(-z,x), \quad \ti z= - z^{-1}, \lb{ch1.8} \\
&A(x)=\begin{pmatrix} \alpha(x)^2 + \kap(x) & - \alpha(x) \\ 
- \alpha(x) & 1 \end{pmatrix} > 0, \quad 
B(x)=\begin{pmatrix} 0 & -1 \\ -1 & 0 \end{pmatrix} = B(x)^*, \quad x\in\bbR.  \no
\end{align}
 We emphasize, in particular, that the new zero-curvature matrix $U(-z, \, \cdot \,)$ (cf.\ \cite[Appendix~A]{Ec15}) 
renders the Hamiltonian system \eqref{ch1.7} linear with respect to the spectral parameter 
$\tilde z$ and hence amenable to standard spectral theory (more precisely, Weyl--Titchmarsh theory and 
all its ramifications; see Section \ref{s4}). In particular, with $(u, \kap)$ subject to conditions 
\eqref{5.3} the Hamiltonian system \eqref{ch1.7} 
is in the limit point case at $x = \pm \infty$. Other known examples of zero-curvature matrices 
$U(-z, \, \cdot \,)$ (e.g., the one employed in \cite{HF15}) lead to Hamiltonian systems quadratic 
in $\tilde z$ and hence their spectral theory cannot be handled by the methods indicated in 
Section \ref{s4}. Upon characterizing certain classes of Nevanlinna--Herglotz functions 
defined in terms of polynomials and their square roots (cf.\ Lemma \ref{l5.2}), we derive in 
detail the half-line Weyl--Titchmarsh functions corresponding to the Hamiltonian system 
\eqref{ch1.7} in connection with the stationary algebro-geometric solutions $(u, \kap)$ 
discussed in Section \ref{s3}. This then enables us to derive the corresponding $2 \times 2$ 
matrix Weyl--Titchmarsh functions and the associated $2 \times 2$ matrix spectral function in the Nevanlinna--Herglotz representation of the former on the entire real line (cf.\ Theorem \ref{t5.3}), 
again in the context of stationary algebro-geometric solutions $(u, \kap)$ of the $\sCH$ 
hierarchy. Here we just remark that these $2 \times 2$ matrix functions are both expressed in 
terms of the polynomials $F_n(z, \, \cdot \,)$, $G_{n+1}(z, \, \cdot \,)$, $H_n(z, \, \cdot \,)$,  
and $R_{2n + 2}(z)$ (cf.\ \eqref{5.21}--\eqref{5.21a}). The limit point (i.e., self-adjointness) 
property of the Hamiltonian system corresponding to real-valued, bounded stationary, algebro-geometric $\CH$ solutions then restricts the motion of the zeros and poles of the fundamental function 
$\phi$ to real intervals (the closure of spectral gaps, cf.\ Theorem \ref{t5.4}). Together with 
the Dubrovin initial value problem treated in Theorem \ref{t5.8}, this finally leads to the 
determination of the isospectral set of all real-valued, smooth and bounded algebro-geometric 
solutions of the stationary $\CH$ equation, $\sCH_n(u,\kap) = 0$, as the real $n$-dimensional 
torus $\bbT^n$ in Corollary \ref{c5.9}. 
      
We focus primarily on the case of stationary $\CH$ hierarchy solutions as the time-dependent 
case subordinates to the stationary one with respect to isospectral torus questions, a fact that 
is briefly commented on at the end of Section \ref{s5}. 

As noted, the special case $\kap \equiv 0$ reduces the two-component Camassa--Holm 
hierarchy, $\CH$, to the standard (i.e., one-component) Camassa--Holm hierarchy ($\ch1$). 
This special case was treated in detail in \cite{ACFHM01}, \cite{AF01}, \cite{GH03}, 
\cite[Ch.~5]{GH03a}. The corresponding isospectral torus of all real-valued, smooth, and 
bounded algebro-geometric solutions of the one-component $\ch1$ hierarchy has been derived 
in \cite{GH08}.

\section{The \texorpdfstring{$\CH$}{CH-2} Hierarchy, Recursion Relations, and Hyperelliptic 
Curves}  \lb{s2}

In this section we review the basic construction of the two-component Camassa--Holm
hierarchy ($\CH$) using an appropriate zero-curvature approach. An alternative approach to the $\CH$ hierarchy was first derived in \cite{HF15}. Both approaches 
follow standard arguments first developed in \cite{GH03} (cf.\ also 
\cite[Ch.~5]{GH03a}). 

Throughout this section we will suppose the following hypothesis.
\begin{hypothesis}\lb{h2.1} 
Suppose that $u, \kap \colon \bbR\to\bbC$. \\
In the stationary case we assume that
\begin{equation}
u, \kap \in C^\infty(\bbR), \; u^{(m)}, \kap^{(m)} \in L^\infty(\bbR), \; 
m\in\bbN_{0}. \lb{ch2.1}
\end{equation}
In the time-dependent case $($cf.\ \eqref{ch2.40}--\eqref{ch2.47}$)$ we
suppose 
\begin{align}
&u(\dott,t), \kap(\dott,t) \in C^\infty(\bbR), \; 
\f{\partial^m u}{\partial x^m}(\dott,t), \f{\partial^m \kap}{\partial x^m}(\dott,t) 
\in L^\infty(\bbR), \; m\in\bbN_{0}, \, t\in\bbR, \no \\
&u(x,\dott), u_{x}(x,\dott), \kap(x,\dott) \in C^1(\bbR), \; x\in\bbR. \lb{ch2.2}
\end{align}
\end{hypothesis}

We start by formulating the basic polynomial setup. One defines
$\{f_\ell\}_{\ell\in\bbN_0}$ recursively by
\begin{align}
& f_{0}=1, \quad f_1 = - 2 u + c_1,   \no \\ 
& f_{\ell,x}=-2\Green\big(2(4u-u_{xx})f_{\ell -1,x}
+(4u_{x}-u_{xxx})f_{\ell -1}      \lb{ch2.3} \\
& \hspace*{18mm} - 2 \kap f_{\ell - 2,x} - \kap_x f_{\ell - 2}\big), \quad 
\ell\in\bbN \backslash \{1\},    \no 
\end{align}
where $c_1$ is an integration constant and $\Green$ is given by
\begin{equation}
\Green\colon L^\infty(\bbR)\to L^\infty(\bbR), \quad
(\Green v)(x)=\frac14 \int_{\bbR} dy\, e^{-2\abs{x-y}} v(y),
\quad x\in\bbR, \; v\in L^\infty(\bbR). \lb{ch2.4} 
\end{equation}
One observes that $\Green$ is the 
resolvent of minus the one-dimensional Laplacian when the spectral parameter is 
equal to $-4$, that is,
\begin{equation}
    \Green=\bigg(-\frac{d^2}{dx^2}+4\bigg)^{-1}. \lb{ch2.5} 
\end{equation}
The coefficients $f_{\ell}$, $\ell \in \bbN$, $\ell \geq 2$, 
are non-local with respect to $u$. At each level a new integration
constant,  denoted by $c_{\ell}$, is introduced.
Moreover, abbreviating 
\begin{equation}
\alpha = u_x + 2 u,    \lb{ch2.7}
\end{equation}
we introduce coefficients  
$\{g_\ell\}_{\ell\in\bbN_0}$ and $\{h_\ell\}_{\ell\in\bbN_0}$ by
\begin{equation}
g_\ell=f_{\ell} + \alpha f_{\ell - 1} + \f12 f_{\ell,x}, \quad 
h_{\ell} = - \big(\alpha^2 + \kap\big)f_{\ell} - g_{\ell+2,x}, 
\quad \ell\in\bbN_{0}, \lb{ch2.8} 
\end{equation}
with the convention $f_{-1}=0$. Explicitly, one computes
\begin{align}
f_{0}&=1, \quad 
f_{1}=-2u+c_{1}, \quad
f_{2}=2u^2+2\Green\big(u_{x}^2+8u^2 + \kap\big)+c_1(-2 u)+c_2, 
\no \\
g_{0}&=1,\quad g_{1} = c_{1},    \no \\
g_{2}&= - 2u^2+2\Green\big(u_x^2
+ u_x u_{xx}+8u u_x+8u^2 + \kap + 2^{-1}\kap_x\big) + c_2,    \lb{ch2.9} \\ 
h_{0} &= -2 \Green\big(16 u u_x+2 u_x^2+2 u_x u_{xx}+16 u^2 + 2^{-1} \kap_{xx}
+ \kap_x\big) + 4 u^2 - \kap, \no\\
& \, \text{ etc.}     \no 
\end{align}
For later use we also note 
\begin{equation} 
h_{\ell,x} - 2h_{\ell} - 2 \alpha h_{\ell-1} - 2\big(\alpha^2 + \kap)g_{\ell}=0, 
    \quad \ell\in\bbN_0, \lb{ch2.10}
\end{equation}
again using the convention $h_{-1}=0$. This can be easily seen by first using \eqref{ch2.8} to eliminate $g_\ell$, $h_\ell$ which
eventually reduces \eqref{ch2.10} to \eqref{ch2.3}.

Given Hypothesis \ref{h2.1}, one introduces the $2\times 2$ matrix $U$ by
\begin{align}
U(z,x)&= - z^{-1} \begin{pmatrix} \alpha(x) & -1\\
\alpha(x)^2 + \kap(x) & - \alpha(x) \end{pmatrix} + \begin{pmatrix}-1 &0\\
0 &1 \end{pmatrix}, \quad z\in\bbC\backslash\{0\}, \; x\in\bbR, \lb{ch2.14} 
\end{align}
and for each $n\in\bbN_{0}$ the following $2\times 2$  matrix $V_n$ by
\begin{equation}
V_{n}(z,x) = z^{-1} 
\begin{pmatrix}-G_{n+1}(z,x)& F_{n}(z,x)\\
 H_{n}(z,x) & G_{n+1}(z,x)
\end{pmatrix}, \quad n\in\bbN_0, \; z\in\bbC\backslash\{0\}, 
\; x\in\bbR, \lb{ch2.15} 
\end{equation} 
assuming $F_n$, $H_n$, and $G_{n+1}$ to be polynomials of degree $n$ 
and $n+1$, respectively, with respect to $z$ and $C^\infty$ in $x$. In addition, we will 
choose $F_n$ and $G_{n+1}$ to be monic in $z$. Postulating 
the zero-curvature condition
\begin{equation}
-V_{n,x}(z,x)+[U(z,x),V_n(z,x)]=0, \lb{ch2.16}
\end{equation}
one finds
\begin{align}
- z F_{n,x}(z,x)- 2 [\alpha(x) + z] F_n(z,x) + 2 G_{n+1}(z,x) &=0, 
\lb{ch2.17} \\ 
-z G_{n+1,x}(z,x) - \big[\alpha(x)^2 + \kap(x)\big] F_{n}(z,x) - H_n(z,x) &=0, \lb{ch2.18} \\ 
-z H_{n,x}(z,x) + 2 [\alpha(x) + z] H_{n}(z,x) 
+ 2\big[\alpha(x)^2 + \kap(x)\big] G_{n+1}(z,x)&=0.
\lb{ch2.19}
\end{align}
In addition, employing \eqref{ch2.17} and \eqref{ch2.18}, one infers that \eqref{ch2.19} is 
equivalent to
\begin{equation}
H_{n,x}(z,x) + 2 [\alpha(x) + z] G_{n+1,x}(z,x) - \big[\alpha(x)^2 + \kap(x)\big] F_{n,x}(z,x) = 0.
\end{equation}
From \eqref{ch2.17}--\eqref{ch2.19} one infers that
\begin{equation}
\f{d}{dx}\det(V_n(z,x))= - z^{-2}\f{d}{dx}\big[G_{n+1}(z,x)^2+F_n(z,x)
H_n(z,x)\big] = 0, \lb{ch2.20}
\end{equation}
and hence
\begin{equation}
G_{n+1}(z,x)^2 + F_n(z,x) H_n(z,x)=R_{2n+2}(z), \lb{ch2.23}
\end{equation}
where $R_{2n+2}$ is an $x$-independent monic polynomial with respect to $z$ of degree 
$2n + 2$ and hence of the form 
\begin{equation} 
R_{2n+2}(z)=\prod_{m=0}^{2n+1}(z-E_m), \quad 
\{E_m\}_{m = 0}^{2n+1} \subset \bbC.   \lb{ch2.22}
\end{equation} 
Using equations \eqref{ch2.17}--\eqref{ch2.19} one can also derive
individual differential equations for $F_n$ and $H_n$. Focusing on
$F_n$ only, one obtains
\begin{align}
\begin{split} 
&F_{n,xxx}(z,x) - 4 F_{n,x} -4\big[z^{-1}(4u(x)-u_{xx}(x)) - z^{-2} \kap\big] F_{n,x}(z,x)    \\
& \quad -2z^{-1} \big[(4u_{x}(x)-u_{xxx}(x)) - z^{-1} \kap_x\big]F_{n}(z,x)=0,    \lb{ch2.37}   
\end{split}  
\end{align}
and
\begin{align}
\begin{split} 
&-(z^2/2)F_{n,xx}(z,x)F_n(z,x)+(z^2/4)F_{n,x}(z,x)^2      \\
& \quad + \big[z^2 + z(4u(x)-u_{xx}(x)) - \kap\big]F_n(z,x)^2=R_{2n+2}(z). 
\lb{ch2.39a}
\end{split} 
\end{align}

Next, we connect the recursion relations \eqref{ch2.3}, \eqref{ch2.8} with the polynomials 
$F_{n}$, $H_{n}$, and $G_{n+1}$ by making the ansatz 
\begin{align}
F_{n}(z,x) &= \sum_{\ell=0}^n f_{n-\ell}(x)z^\ell, \quad
G_{n+1}(z,x)=\sum_{\ell=0}^{n+1}g_{n+1-\ell}(x)z^\ell - f_{n+1} - \f12 f_{n+1,x},  \no \\
H_{n}(z,x) &= \sum_{\ell=0}^n h_{n-\ell}(x)z^\ell + g_{n+2,x}.     \lb{ch2.26} 
\end{align} 
Inserting the ansatz \eqref{ch2.26} into \eqref{ch2.17} and comparing coefficients shows that this equation holds due to \eqref{ch2.8}.
Similarly, inserting \eqref{ch2.26} into \eqref{ch2.18} shows that the latter equation holds due to 
\eqref{ch2.8} and $g_0'=g_1'=0$ if and only if the term linear in $z$ vanishes, 
\begin{equation}\label{ch2.27}
f_{n+1,x}+\frac{1}{2} f_{n+1,xx} = 0.
\end{equation} 
Finally, inserting \eqref{ch2.26} into \eqref{ch2.19} all coefficients of $z^\ell$ for $\ell\ge2$ 
cancel due to \eqref{ch2.10}. For the constant (i.e., $z^0$) term one gets, using \eqref{ch2.8},
\begin{equation}
2 \alpha \left(g_{n+2,x}+h_n\right) + 2\left(\alpha^2+\kap\right) \left(g_{n+1} - f_{n+1} -\frac12 f_{n+1,x}\right) =0
\end{equation}
Similarly, for the $z^1$-term one gets using \eqref{ch2.10},  \eqref{ch2.8}, and \eqref{ch2.3} 
(in this order), 
\begin{align}
&h_{n,x} + g_{n+2,xx} - 2\alpha h_{n-1} - 2 \left(g_{n+2,x}+h_{n}\right)
-2 \left(\alpha^2+\kap\right) g_{n}   \no \\
& \quad = - 2 g_{n+2,x}+g_{n+2,xx}   \no \\ \label{ch2.28}
& \quad = - \kap_x f_n - 2\kap f_{n,x} + 2\alpha \left(f_{n+1}+ \frac12 f_{n+1,x}\right)_x.
\end{align}
Hence \eqref{ch2.19} holds if and only if the final right-hand side of \eqref{ch2.28} vanishes.

In summary, the zero-curvature condition \eqref{ch2.16} will hold if and only if
\begin{equation}
\bigg(f_{n+1,x}+ \frac12 f_{n+1,xx}\bigg)=0 \, \text{ and } \, 
\kap_x f_n + 2 \kap f_{n,x} = 0.     \lb{ch2.24} 
\end{equation}
For reasons to become clear in connection with the time-dependent formulation, we will replace
the first equation in \eqref{ch2.24} by the equivalent one 
\begin{equation} 
\bigg(\frac{d}{dx}+2\bigg)^{-1} \bigg(f_{n+1}+ \frac12 f_{n+1,x}\bigg)_x = \frac{1}{2} f_{n+1,x} = 0.
\end{equation}
Thus, the zero-curvature condition \eqref{ch2.16} is equivalent to 
\begin{equation}
\sCH_n(u, \kap) = \begin{pmatrix} 
 \frac{1}{2} f_{n+1,x} \\[1mm] 
- \kap_x f_n - 2 \kap f_{n,x}
\end{pmatrix}  = 0, \quad n \in \bbN_0.    \lb{ch2.29a}
\end{equation}
Varying $n\in\bbN_0$ in \eqref{ch2.29a} then defines the stationary $\CH$ hierarchy.

We record the first two equations explicitly,
\begin{align}
\sCH_0(u, \kap)&= \begin{pmatrix} -u_{x}  \\[1mm] 
 - \kap_x \end{pmatrix} = 0,  \no \\
\sCH_1(u, \kap)&= \begin{pmatrix} \Green (2u_x u_{xx} +16u u_{x} + \kap_x) + 2u u_{x} -c_{1} u_{x} \\[1mm]  
2 \kap_x u + 4 \kap u_x + c_1 (-\kap_x) \end{pmatrix} =0, \no \\ 
& \hspace*{-15mm} \text{etc.}    \lb{ch2.30} 
\end{align}

By definition, the set of solutions of \eqref{ch2.29a}, with $n$ ranging in
$\bbN_0$, represents the class of algebro-geometric $\CH$ 
solutions. If $(u, \kap)$ satisfies one of the stationary $\CH$ equations in
\eqref{ch2.29a} for a particular value of $n$, then it satisfies 
infinitely many
such equations of order higher than $n$ for certain choices of integration
constants $c_\ell$ (see \cite[Remark 1.5]{GH03a} for the corresponding argument for the KdV equation). 

Next, we turn to the time-dependent $\CH$ hierarchy. Introducing a    
deformation parameter $t_n\in\bbR$ into $u$ (i.e., replacing 
$(u(x), \kap(x))$ by $(u(x,t_n), \kap(x,t_n))$), the definitions \eqref{ch2.14}, 
\eqref{ch2.15}, and \eqref{ch2.26} of $U$, $V_n$, and 
$F_n$, $G_{n+1}$,
and $H_n$, respectively, still apply. The corresponding zero-curvature 
relation reads
\begin{equation}
U_{t_n}(z,x,t_n)-V_{n,x}(z,x,t_n)+[U(z,x,t_n),V_n(z,x,t_n)]=0, 
\quad n\in\bbN_0,\lb{ch2.40}
\end{equation}
which results in the following set of time-dependent equations
\begin{align}
&z F_{n,x}(z,x,t_n) = - 2 [\alpha(x,t_n) +z] F_n(z,x,t_n) + 2 G_{n+1}(z,x,t_n),  
\lb{ch2.42} \\
&z \alpha_{t_n}(x,t_n) = z G_{n+1,x}(z,x,t_n) 
+ \big[\alpha(x,t_n)^2 + \kap(x,t_n)\big] F_n(z,x,t_n) + H_{n}(z,x,t_n),    \lb{ch2.43} \\ 
&z [2 \alpha(x,t_n) \alpha_{t_n}(x,t_n) + \kap_{t_n}(x,t_n)] 
= - z H_{n,x}(z,x,t_{n})     \lb{ch2.41} \\
&\quad + 2 [\alpha(x,t_n) + z] H_{n}(z,x,t_{n}) +  
2 \big[\alpha(x,t_n)^2 + \kap(x,t_n)\big] G_{n+1}(z,x,t_{n})=0.     \no 
\end{align}
Now one proceeds as in the stationary case to conclude that these equations hold
if and only if
\begin{equation}
\alpha_{t_n} + f_{n+1,x}+\frac{1}{2} f_{n+1,xx} = 0
\end{equation} 
and
\begin{equation}
2\alpha \alpha_{t_n} + \kap_{t_n} - \kap_x f_n - 2 \kap f_{n,x}  
+ 2\alpha \left(f_{n+1}+ \frac12 f_{n+1,x}\right)_x =0.
\end{equation} 
Hence one arrives at the corresponding time-dependent hierarchy 
\begin{align}
\CH_n(u, \kap) = \begin{pmatrix} 
u_{t_n} + \frac12 f_{n+1,x} \\[1mm] 
\kap_{t_n}  - \kap_x f_n - 2\kap f_{n,x} 
\end{pmatrix}    
= 0, \quad n \in \bbN_0.    \lb{ch2.46}
\end{align}
Varying $n\in\bbN_0$ in \eqref{ch2.46} then defines the time-dependent 
$\CH$ hierarchy. We record the first few equations explicitly,
\begin{align}
\CH_0(u, \kap)&= \begin{pmatrix} u_{t_{0}}-u_{x} \\[1mm] 
\kap_{t_0} - \kap_{x} \end{pmatrix} = 0, \no \\
\CH_1(u, \kap)&= \begin{pmatrix} u_{t_{1}} +\Green (2u_x u_{xx} +16u u_{x} + \kap_x) + 2u u_{x} -c_{1} u_{x} \\[1mm] 
\kap_{t_1} + 2 \kap_x u + 4 \kap u_x + c_1 (-\kap_x) \end{pmatrix} 
=0,   \no \\
& \hspace*{-15mm} \text{etc.}    \lb{ch2.47} 
\end{align}
Up to an inessential scaling of the $(x,t_1)$ variables,
${\CH}_1(u)=0$ with $c_1=0$ represents {\it the two-component Camassa--Holm} equation as discussed, for instance in \cite{HI10}, \cite{HF15}.
In this respect we remark that the first component is more frequently written in the literature as
\begin{align}\no
\Green^{-1} \left( u_{t_n} +  \frac12 f_{n+1,x} \right) &= 4u_{t_n}-u_{xxt_n}+(u_{xxx}-4u_{x})f_{n} +2(u_{xx} - 4 u)f_{n,x}   \no \\
& \quad + \kap_x f_{n-1} + 2 \kap f_{n-1,x}, \quad n \in \bbN.
\end{align}

\section{The Stationary Algebro-Geometric \texorpdfstring{$\CH$}{CH-2} Formalism} \lb{s3}

This section is devoted to a quick review of the stationary $\CH$ 
hierarchies and the corresponding algebro-geometric formalism. This topic has 
first been discussed in \cite{HF15} using a different zero-curvature pair $(U, V_n)$. 
These approaches are standard and follow the lines developed in \cite{GH03} 
(cf.\ also \cite[Ch.~5]{GH03a}). 

We start with the stationary hierarchy and hence impose the following assumptions: 

\begin{hypothesis} \lb{h3.0}
Suppose that $u, \kap \colon \bbR\to\bbC$ satisfy
\begin{equation}
u, \kap \in C^\infty(\bbR), \; u^{(m)}, \kap^{(m)} \in L^\infty(\bbR), \; 
m\in\bbN_{0},    \label{ch3.0}
\end{equation}
and let all associated quantities \eqref{ch2.3}, \eqref{ch2.8}, \eqref{ch2.26} be defined
as in the previous section. Moreover, suppose (cf.\ \eqref{ch2.23}, \eqref{ch2.22})
\begin{equation}
\{E_m\}_{m = 0}^{2n+1} \subset \bbC \backslash \{0\}.   \lb{ch3.0a}
\end{equation}
\end{hypothesis}

Recalling \eqref{ch2.22},  
\begin{equation}
R_{2n+2}(z)=\prod_{m=0}^{2n+1}(z-E_m),       \lb{ch3.1}
\end{equation}
we introduce the (possibly singular)
hyperelliptic  curve $\calK_n$ of arithmetic genus $n$ defined by
\begin{equation}
\calK_n \colon \calF_n(z,y)=y^2-R_{2n+2}(z)=0. \lb{ch3.3}
\end{equation} 
We compactify $\calK_n$ by adding two points at infinity,  $\Pinfp$, 
$\Pinfm$, with  $\Pinfp\neq \Pinfm$, still denoting
its projective closure by  $\calK_n$.  Hence $\calK_n$
becomes a two-sheeted Riemann surface
of arithmetic genus $n$.  Points $P$ on $\calK_n\backslash\{\Pinfpm\}$ are
denoted by $P=(z,y)$, where $y(\dott)$ denotes the meromorphic
function on $\calK_n$ satisfying $\calF_n(z,y)=0$. 

For notational simplicity we will usually tacitly assume that $n\in\bbN$
(the case $n=0$ being trivial).

In the following the roots of the polynomials $F_n$ and $H_n$ will play a
special role and hence we introduce on $\bbC\times\bbR$
\begin{equation}
F_n(z,x)=\prod_{j=1}^n [z-\mu_j(x)], \quad
H_n(z,x)=h_{0}(x)\prod_{j=1}^n [z-\nu_j(x)], \lb{ch3.5}
\end{equation}
temporarily assuming
\begin{equation}
h_0(x)\neq 0, \quad x\in\bbR. \lb{ch3.5a}
\end{equation}
Moreover, we introduce
\begin{align}
\hat\mu_j(x)&=(\mu_j(x), -G_{n+1}(\mu_j(x),x))\in\calK_n,
& j & =1,\dots,n, \; x\in\bbR, \lb{ch3.6}\\
\hat\nu_j(x)&=(\nu_j(x), G_{n+1}(\nu_j(x),x))\in\calK_n,
& j & =1,\dots,n, \; x\in\bbR. \lb{ch3.7}
\end{align}
The branch of $y(\dott)$ near $\Pinfpm$ is fixed according to
\begin{equation}
\f{y(P)}{z(P)^{n+1}}
\underset{\substack{\abs{z(P)}\to\infty\\P\to\Pinfpm}}{=} \mp \big[1 + c_1(\ul E) z(P)^{-1} + \Oh\big(z(P)^{-2}\big)\big].   \lb{ch3.8a}
\end{equation}
Due to assumption \eqref{ch3.0}, $u$ is smooth and bounded, and hence  
$F_n(z,\dott)$ and $H_{n}(z,\dott)$ share the same property. 
Thus, one concludes 
\begin{equation}
\mu_j, \nu_k \in C(\bbR), \; j,k=1,\dots,n,  
\lb{ch3.9}
\end{equation}
taking multiplicities (and appropriate reordering) of the zeros of $F_n$ 
and $H_n$ into account. 

Equation \eqref{ch2.39a} leads to an explicit determination
of the integration constants $c_1,\dots,c_{n}$ in the stationary $\CH$ 
equations \eqref{ch2.29a} in terms of the zeros $E_m$, $m=0,\dots, 2n+1$, 
of the associated polynomial $R_{2n+2}$ in
\eqref{ch2.22}, as follows: Choosing $P=(z,y) \in \Pi_{n,+}$ (cf.\ \eqref{5.13}, \eqref{5.14}) and 
inserting 
\begin{equation}\label{ch2.31}
\frac{F_{n}(z,x)}{y(P)} = - \sum_{\ell=0}^\infty \hat{f}_\ell(x) z^{-\ell-1}
\end{equation}
into \eqref{ch2.39a}, one obtains the nonlinear recursion
\begin{align}
 \hat{f}_0 =& 1, \quad \hat{f}_1 = - 2u,   \no\\
 \hat{f}_\ell =& -\Green\Bigg(
 \sum_{m=1}^{\ell-1} \left[ \frac{1}{2} \hat{f}_{m,x} \hat{f}_{\ell-m,x} 
 + \hat{f}_m \big(2 \hat{f}_{\ell-m}- \hat{f}_{\ell-m,xx}\big)\right] \\
   & \qquad \; {} +2 \sum_{m=0}^{\ell-1} \hat{f}_m \left[(4u-u_{xx}) \hat{f}_{\ell-m-1} 
  - \kap \hat{f}_{\ell-m-2}\right] \Bigg), \quad \ell \in \bbN \backslash\{1\}.    \no
\end{align}
Furthermore, inserting \eqref{ch2.31} into \eqref{ch2.37} one sees that  $\hat{f}_\ell$ also 
satisfies \eqref{ch2.3}, and by homogeneity considerations one
infers
\begin{equation}
f_\ell = \sum_{m=0}^\ell c_{\ell - m} \hat{f}_m.
\end{equation}
Using again \eqref{ch2.31} and \eqref{ch2.26} one finally obtains 
\begin{equation}
c_\ell=c_\ell(\ul E), \quad \ell=0,\dots,n, \lb{ch2.39C}
\end{equation}
where $c_k(\ul E)$, $k \in \bbN_0$, denote the asymptotic expansion coefficients of 
$y(P)^{-1} = -\sum_{\ell=0}^\infty c_\ell(\ul E) z^{-n-\ell-1}$.
Explicitly (cf.\ \cite[App.~D]{GH03}), 
\begin{align}
c_0(\ul E)&=1,  \quad c_1(\ul E) = - \f{1}{2} \sum_{m=0}^{2n+1} E_m,    \no \\
c_k(\ul E)&=-\!\!\!\!\!\sum_{\substack{j_1,\dots,j_{2n+1}=0\\
 j_1+\cdots+j_{2n+1}=k}}^{k}\!\!
\f{(2j_1)!\cdots(2j_{2n+1})!}
{2^{2k} (j_1!)^2\cdots (j_{2n+1}!)^2 (2j_1-1)\cdots(2j_{2n+1}-1)} \no \\
& \hspace*{2.4cm} \times E_1^{j_1}\cdots E_{2n+1}^{j_{2n+1}}, \quad 
k\in\bbN. \label{ch2.39D}
\end{align}

Next, we introduce the fundamental meromorphic function $\phi(\dott,x)$ on
$\calK_n$ by
\begin{equation}
\phi(P,x)=\f{y - G_{n+1}(z,x)}{F_n(z,x)} 
=\f{H_{n}(z,x)}{y + G_{n+1}(z,x)}, \quad
 P=(z,y)\in\calK_n, \, x\in\bbR. \lb{ch3.11}
\end{equation}
Assuming \eqref{ch3.5a}, the divisor $(\phi(\dott,x))$
of $\phi(\dott,x)$ is given by
\begin{equation}
(\phi(\dott,x))=\calD_{\Pinfm\hat{\ul\nu}(x)}
-\calD_{\Pinfp\hat{\ul\mu}(x)}, \lb{ch3.13}
\end{equation} 
taking into account \eqref{ch3.8a}. 
Here we abbreviated
\begin{equation}
\hat{\ul\mu}=\{\hat\mu_1,\dots,\hat\mu_n\}, \, 
\hat{\ul\nu}=\{\hat\nu_1,\dots,\hat\nu_n\} \in \sigma^n\calK_n, \lb{ch3.14}
\end{equation}
where $\sigma^m\calK_n$, $m\in\bbN$, denotes the $m$th symmetric product 
of $\calK_n$. Moreover, we used the following convenient notation for a positive divisor $\calD_{\ul Q}$ 
of degree $n$ on $\calK_n$, 
\begin{align}
&{\calD}_{\ul{Q}} \colon {\calK}_n \to \bbN_0,  \quad
 P \mapsto  {\calD}_{\ul{Q}}(P)=\begin{cases} m & \text{if 
$P$ occurs $m$ times in  $\{Q_1,\dots, Q_n\}$}, \\
0& \text{if $P\notin\{Q_1,\dots, Q_n\}$},
\end{cases} \no \\
& \hspace*{7cm} \ul{Q}=\{Q_1,\dots, Q_n\}\in\sigma^n\calK_n,  \label{aa53}
\end{align}
and used the following notation for divisors of degree $n + 1$ on $\calK_n$, 
\begin{equation}
\calD_{Q_0\ul Q}=\calD_{Q_0}+\calD_{\ul Q}, \quad Q_0\in\calK_n, 
\end{equation}
where for any $Q\in\calK_n$,
\begin{equation}
{\calD}_{Q} \colon {\calK}_n \to \bbN_0, \quad
P \mapsto  \calD_Q(P)= \begin{cases} 1 & \text{for $P=Q$},\\
0 & \text{for $P\in\calK_n \backslash \{Q\}$}.
\end{cases} \label{aa61A}
\end{equation}

If $h_0$ is permitted to vanish at a point $x_1\in\bbR$, then for $x=x_1$, the polynomial $H_n(\dott,x_1)$ is at most of degree $n-1$ (cf.\ \eqref{ch2.26}). 
Since this can be viewed as a limiting case of \eqref{ch3.13}, we will henceforth 
not particularly distinguish the case $h_0\neq 0$ from the more general situation 
where $h_0$ is permitted to vanish.

Given $\phi(\dott,x)$, one defines the associated Baker--Akhiezer vector 
$\Psi(\dott,x,x_0)$ on $\calK_n\backslash\{\Pinfp,\Pinfm\}$ by
\begin{equation}
\Psi(P,x,x_0)=\begin{pmatrix} \psi_1(P,x,x_0) \\ \psi_2(P,x,x_0)
\end{pmatrix}, \quad
P\in\calK_n\backslash\{\Pinfp,\Pinfm\}, \; (x,x_0)\in\bbR^2,
\lb{ch3.15}
\end{equation}
where
\begin{align}
\psi_1(P,x,x_0)&=\exp\left(- z^{-1} \int_{x_0}^x dx'\,
\phi(P,x')-(x-x_0) - z^{-1} \int_{x_0}^x dx' \, \alpha(x')\right), \lb{ch3.16}\\
\psi_2(P,x,x_0)&=-\psi_1(P,x,x_0) \phi(P,x). \lb{ch3.17}
\end{align} 

The basic properties of $\phi$ and $\Psi$ then read as follows.

\begin{lemma} \lb{l3.1}
Assume Hypothesis \ref{h3.0} and that the $n$th stationary $\CH$ 
equation \eqref{ch2.29a} holds on some open interval $\Omega \subseteq \bbR$. 
Moreover, suppose that $P=(z,y)\in\calK_n\backslash \{\Pinfp,\Pinfm\}$, 
$(x,x_0)\in \Omega^2$. Then $\phi$ satisfies the Riccati-type equation
\begin{equation}
\phi_x(P,x) - z^{-1}\phi(P,x)^2 - 2z^{-1}(\alpha(x) + z)\phi(P,x) 
- 2 z^{-1}[\alpha(x)^2 + \kap(x)] = 0, 
\lb{ch3.18}
\end{equation}
as well as
\begin{align}
\phi(P,x)\phi(P^*,x)&=-\f{H_{n}(z,x)}{F_n(z,x)}, \lb{ch3.19}\\
\phi(P,x)+\phi(P^*,x)&=-2\f{G_{n+1}(z,x)}{F_n(z,x)}, \lb{ch3.20}\\
\phi(P,x)-\phi(P^*,x)&=\f{2y}{F_n(z,x)}, \lb{ch3.21}
\end{align}
while $\Psi$ fulfills
\begin{align}
&\Psi_x(P,x,x_0)=U(z,x)\Psi(P,x,x_0), \lb{ch3.22} \\
&-y \Psi(P,x,x_0)=zV_n(z,x)\Psi(P,x,x_0), \lb{ch3.23} \\
&\psi_1(P,x,x_0)=\bigg(\f{F_n(z,x)}{F_n(z,x_0)}\bigg)^{1/2}
\exp\bigg(-(y/z)\int_{x_0}^x dx'F_n(z,x')^{-1} \bigg), \lb{ch3.24} \\
&\psi_1(P,x,x_0)\psi_1(P^*,x,x_0)=\f{F_n(z,x)}{F_n(z,x_0)},\lb{ch3.25} \\
&\psi_2(P,x,x_0)\psi_2(P^*,x,x_0)=-\f{H_n(z,x)}{F_n(z,x_0)},\lb{ch3.26}
\\
&\psi_1(P,x,x_0)\psi_2(P^*,x,x_0)+\psi_1(P^*,x,x_0)\psi_2(P,x,x_0)
=2\f{G_{n+1}(z,x)}{F_n(z,x_0)}, \lb{ch3.27} \\
&\psi_1(P,x,x_0)\psi_2(P^*,x,x_0)-\psi_1(P^*,x,x_0)\psi_2(P,x,x_0)
=\f{2y}{F_n(z,x_0)}. \lb{ch3.28}
\end{align}
In addition, as long as the zeros of $F_n(\dott,x)$ are all simple for 
$x\in\Omega$, $\Psi(\dott,x,x_0)$, 
$x,x_0\in\Omega$, is meromorphic on $\calK_n$. 
\end{lemma}
\begin{proof}
The proof of Lemma \ref{l3.1} is standard and follows that of \cite[Lemma~3.1]{GH03} line by line (cf.\ also \cite[Lemma~5.2]{GH03a}). In particular, 
\eqref{ch3.19}--\eqref{ch3.21} are clear from the definition \eqref{ch3.11} of $\phi$ 
and from the fact that $y(P^*) = - y(P)$, 
similarly, \eqref{ch3.22}--\eqref{ch3.28} are immediate from \eqref{ch3.16}, \eqref{ch3.17}, and \eqref{ch3.19}--\eqref{ch3.21}. The Riccati-type equation \eqref{ch3.18} follows from \eqref{ch3.17}, \eqref{ch3.18}, and \eqref{ch3.18}.  
Meromorphy of $\Psi(\dott,x,x_0)$, on $\calK_n$ as long as the zeros of $F_n(\dott,x)$ are all simple follows from 
\begin{equation}
-\f{1}{z}\phi(P,x') \underset{P\to\hat \mu_j(x')}{=} \frac{\partial}{\partial x'}
\ln (F_n(z,x')) + \Oh(1)\text{  as $z\to\mu_j(x')$}, \lb{ch3.29}
\end{equation}
(cf.\ \eqref{ch2.17}, \eqref{ch3.6}, and \eqref{ch3.11}) and \eqref{ch3.16}. 
\end{proof} 

Next, we recall the Dubrovin-type equations for $\mu_j$. In the
remainder of this section we will frequently assume that $\calK_n$ has a 
nonsingular affine part, that is, we suppose that 
\begin{equation}
E_m\in\bbC \backslash \{0\}, \; E_m\neq E_{m'} \text{ for } 
m\neq m', \, m,m'=0,\dots,2n+1. \lb{ch3.30}
\end{equation}

\begin{lemma} \lb{l3.2}  
Assume Hypothesis \ref{h3.0} and that the $n$th stationary $\CH$ equation
$\eqref{ch2.29a}$ holds subject to the constraint $\eqref{ch3.30}$ on an open
interval $\wti\Omega_\mu\subseteq\bbR$. Moreover,  suppose that the zeros
$\mu_j$, $j=1,\dots,n$, of $F_n(\dott)$ remain distinct and nonzero on
$\wti\Omega_\mu$. Then $\{\hat\mu_j\}_{j=1,\dots,n}$, defined by
\eqref{ch3.6},  satisfies the following first-order system of differential
equations
\begin{equation}
\mu_{j,x}(x)=2\f{y(\hat\mu_j(x))}{\mu_{j}(x)}
\prod_{\substack{\ell=1\\ \ell\neq j}}^n  [\mu_j(x)-\mu_\ell(x)]^{-1},
\quad j=1, \dots, n, \, x\in \wti\Omega_\mu. \lb{ch3.31}
\end{equation}
Next, assume the affine part of $\calK_n$ to be nonsingular and introduce the initial
condition
\begin{equation}
\{\hat\mu_j(x_0)\}_{j=1,\dots,n}\subset\calK_n \lb{ch3.32}
\end{equation}
for some $x_0\in\bbR$, where $\mu_j(x_0)\neq 0$, $j=1,\dots,n$, are
assumed to be distinct. Then there exists an open interval
$\Omega_\mu\subseteq\bbR$, with $x_0\in\Omega_\mu$, such that the initial
value problem \eqref{ch3.31}, \eqref{ch3.32} has a unique solution
$\{\hat\mu_j\}_{j=1,\dots,n}\subset\calK_n$ satisfying
\begin{equation}
\hat\mu_j\in C^\infty(\Omega_\mu,\calK_n),\quad j=1, \dots, n,
\lb{ch3.33}
\end{equation}
and $\mu_j$, $j=1,\dots,n$, remain distinct and nonzero on $\Omega_\mu$.
\end{lemma}
\begin{proof}
Since $y(\hat \mu_j) = - G_{n+1} (\mu_j) = -(\mu_j/2) F_{n,x}(\mu_j)$ by 
\eqref{ch2.17} and \eqref{ch3.6}, one computes 
\begin{equation}
F_{n,x}(\mu_j) = - \mu_{j,x} \prod_{\substack{\ell = 1 \\ \ell \neq j}}^n 
(\mu_j - \mu_{\ell}) = - (2/\mu_j) y (\hat \mu_j), \quad j=1, \dots, n, 
\end{equation}
from which the rest follows by standard arguments (cf.\ \cite[Lemma~3.2]{GH03},  \cite[Lemma~5.3]{GH03a}). 
\end{proof}

Combining the polynomial approach in Section \ref{s2} with
\eqref{ch3.5} yields trace formulas for the $\CH$ invariants.
For simplicity we just record two simple cases.

\begin{lemma} \lb{l3.3}
Assume Hypothesis \ref{h3.0} and that the $n$th stationary $\CH$ 
equation \eqref{ch2.29a} holds on some set $\Omega_{\mu}$ as in 
Lemma \ref{l3.2}, and let $x\in \Omega_{\mu}$. Then 
\begin{align}
u(x)&=\f12\sum_{j=1}^n\mu_{j}(x)-\f14\sum_{m=0}^{2n+1} E_{m}, 
\lb{ch3.36} \\
\kap(x)&=-\Bigg(\prod_{m=0}^{2n+1}
E_m\Bigg)\Bigg( \prod_{j=1}^n \mu_j(x)^{-2}\Bigg). \lb{ch3.36a}
\end{align}
\end{lemma}
\begin{proof}
For the proof of Lemma \ref{l3.3} one can follow \cite[Lemma~3.3]{GH03} 
(equivalently,  \cite[Lemma~5.4]{GH03a}) line by line. Indeed, 
\begin{equation}
f_1 = - 2u + c_1, \quad f_1 = - \sum_{j=1}^n \mu_j
\end{equation} 
(cf.\ \eqref{ch2.9} and \eqref{ch3.5}), and 
\begin{equation}
c_1 = - 2^{-1} \sum_{m=0}^{2n+1} E_m 
\end{equation}
(cf.\ \eqref{ch2.39D}), prove \eqref{ch3.36}. Combining 
\begin{equation}
f_n = (-1)^n \prod_{j=1}^n \mu_j, \quad g_{n+1}-f_{n+1}-\frac{1}{2} f_{n,x} = \alpha f_n, \quad 
h_n  + g_{n+2,x}= - \big(\alpha^2 + \kap\big) f_n
\end{equation}
(cf.\ \eqref{ch2.8} and \eqref{ch3.5}) with 
\begin{align}
\begin{split} 
& \big[g_{n+1} - f_{n+1} -2^{-1} f_{n+1,x}\big]^2 + f_n [h_n + g_{n+2,x}]    \\
& \quad = \alpha^2 f_n^2 - \big(\alpha^2 + \kap\big) f_n^2 = \prod_{m=0}^{2n+1} E_m 
\end{split} 
\end{align}
(cf.\ \eqref{ch2.23} and \eqref{ch2.26}), prove \eqref{ch3.36a}. By 
Lemma \ref{l3.2} one concludes that $\mu_j(x) \neq 0$ for all $j=1,\dots,n$, 
$x \in \Omega_{\mu}$.
\end{proof}

One notes that both, $u$ and $\kap$, are uniquely determined by $\mu_j$, 
$j =1,\dots,n$. Moreover, $\kap \to 0$ if some $E_m \to 0$, hence we excluded 
the latter situation. 

\begin{remark} \lb{r3.3} 
The trace (actually, product) formula for $\kap$ in \eqref{ch3.36a} is somewhat familiar from 
the $\ch1$ context where $\kap \equiv 0$. Indeed, combining relations (2.28), (2.29), and (3.7) in 
\cite{GH03} yields
\begin{equation}
4u - u_{xx} = -\Bigg(\prod_{m=0}^{2n+1}
E_m\Bigg)\Bigg( \prod_{j=1}^n \mu_j(x)^{-2}\Bigg),
\end{equation}
an identity derived earlier in the periodic context in \cite{CM99}. 
\end{remark}

Next we turn to asymptotic properties of $\phi$ and $\psi_j$, $j=1,2$.

\begin{lemma} \lb{l3.4}
Assume Hypothesis \ref{h3.0} and assume that the $n$th stationary $\CH$ 
equation \eqref{ch2.29a} holds on some open interval $\Omega \subseteq \bbR$. 
In addition, let
$P=(z,y)\in\calK_n\backslash\{\Pinfp\}$, $x\in\Omega$. Then
\begin{equation}
\phi(P,x)\underset{\zeta\to 0}{=}\begin{cases} 
-2\zeta^{-1} +[-4u(x) + c_1]+\Oh(\zeta), & P\to\Pinfp, \\
\Oh(\zeta), & P\to\Pinfm, \end{cases} 
\quad \zeta=z^{-1},  \lb{ch3.38a} 
\end{equation}
and
\begin{align}
&\psi_1(P,x,x_0)\underset{\zeta\to 0}{=}\exp(\pm(x-x_0))(1+\Oh(\zeta)),
\quad P\to\Pinfpm, \quad \zeta=z^{-1}, \lb{ch3.38c} \\
&\psi_2(P,x,x_0)\underset{\zeta\to 0}{=}\exp(\pm(x-x_0))
 \begin{cases}  -2 \zeta^{-1} +\Oh(1), & P\to\Pinfp, \\
\Oh(\zeta), & P\to\Pinfm, \end{cases} \quad \zeta=z^{-1}. 
\lb{ch3.38d}   
\end{align} 
\end{lemma}
\begin{proof}
This is an immediate consequence of \eqref{ch3.8a}, \eqref{ch3.11}, \eqref{ch3.13}, \eqref{ch3.16}, 
and \eqref{ch3.17}.
\end{proof}

Since the representations of $\phi$ and $u$ in
terms of the Riemann theta function associated with $\calK_n$ (assuming the 
affine part of $\calK_n$ to be nonsingular) are not explicitly needed in this paper 
(yet can be derived as in \cite{GH03} and \cite[Ch.~5]{GH03a}), we
omit the corresponding details. We note that reference \cite{HF15} derives these representations adapted to their framework. 

Finally, we note that solvability of the Dubrovin equations 
\eqref{ch3.31} on $\Omega_\mu\subseteq\bbR$ in fact yields the $n$th 
stationary $\CH$ equation \eqref{ch2.29a}  on $\Omega_\mu$. 

\begin{theorem}\lb{t3.10}
Fix $n\in\bbN$, assume \eqref{ch3.30}, and suppose that
$\{\hat\mu_j\}_{j=1,\dots,n}$ satisfies the stationary Dubrovin equations
\eqref{ch3.31} on an open interval $\Omega_\mu\subseteq\bbR$ such that
$\mu_j$, $j=1,\dots,n$, remain distinct and nonzero on $\Omega_\mu$. Then 
$u, \kap \in C^\infty(\Omega_\mu)$ defined by 
\begin{equation}
u(x)=\f12\sum_{j=1}^n\mu_j(x)-\f{1}{4}\sum_{m=0}^{2n+1} E_m, 
\quad \kap(x) = - \Bigg(\prod_{m=0}^{2n+1}
E_m\Bigg)\Bigg( \prod_{j=1}^n \mu_j(x)^{-2}\Bigg),    \lb{ch3.40}
\end{equation}
satisfy the $n$th stationary $\CH$ equation \eqref{ch2.29a}, that is,
\begin{equation}
\sCH_n(u, \kap)=0 \text{  on $\Omega_\mu$.} \lb{ch3.41}
\end{equation}
\end{theorem}
\begin{proof}
Given the solutions $\hat\mu_j=(\mu_j,y(\hat\mu_j))\in
C^\infty(\Omega_\mu, \calK_n)$, $j=1,\dots,n$ of \eqref{ch3.31} we
introduce 
\begin{align}
F_n(z)&=\prod_{j=1}^n (z-\mu_j), \lb{ch3.42}\\
G_{n+1}(z)&= (\alpha + z) F_n(z)+ (z/2) F_{n,x}(z) \lb{ch3.43}
\end{align}
on $\bbC\times\Omega_\mu$.\ The Dubrovin equations imply
\begin{equation}
y(\hat\mu_j)=\f12\mu_j \mu_{j,x}\prod_{\substack{\ell=1\\ \ell\neq
j}}^n(\mu_j-\mu_\ell)=-\f12\mu_j F_{n,x}(\mu_j)= - G_{n+1}(\mu_j).
 \lb{ch3.44}
\end{equation}
Thus, 
\begin{equation}
R_{2n+2}(\mu_j) - G_{n+1}(\mu_j)^2=0, \quad j=1,\dots,n, \lb{ch3.45}
\end{equation}
and one can write 
\begin{equation} 
R_{2n+2}(z) - G_{n+1}(z)^2 = F_n(z) H(z),    \lb{ch3.54} 
\end{equation}
for some polynomial $H$ with respect to $z$. Investigating the leading asymptotics of 
$H$ as $|z| \to \infty$ reveals that the degree of $H$ equals at most $n$ and we thus write $H=H_n$ 
from now on. Indeed, one computes (for $n \in \bbN$, $n \geq 2$, and analogously for $n =0,1$), 
\begin{align}
& G_{n+1}(z) = z^{n+1} + \bigg[- \f{1}{2} \sum_{j=1}^n \mu_{j,x} - \sum_{j=1}^n \mu_j + \alpha \bigg] z^n    
\no \\
& \qquad + \Bigg[\f{1}{2} \sum_{\substack{j_1,j_2=1\\ j_1 < j_2}}^n [\mu_{j_1} \mu_{j_2,x} 
+ \mu_{j_1,x} \mu_{j_2} + 2 \mu_{j_1}\mu_{j_2}] - \alpha \sum_{j=1}^n \mu_j\Bigg] z^{n-1} 
+ \Oh\big(z^{n-2}\big)   \no \\
& \quad = z^{n+1} - \f{1}{2} \bigg(\sum_{m=0}^{2n+1} E_m\bigg) z^n   \lb{ch3.55} \\
& \qquad + \Bigg[\f{1}{2} \sum_{\substack{j_1,j_2=1\\ j_1 < j_2}}^n [\mu_{j_1} \mu_{j_2,x} 
+ \mu_{j_1,x} \mu_{j_2} + 2 \mu_{j_1}\mu_{j_2}] - \alpha \sum_{j=1}^n \mu_j\Bigg] z^{n-1} 
+ \Oh\big(z^{n-2}\big),  \no 
\end{align}
where we used $\alpha = u_x + 2u$ and the trace formula for $u$ (and hence for $u_x$) in \eqref{ch3.40}. Insertion of \eqref{ch3.55} into \eqref{ch3.54} confirms that $H$ has degree at most 
$n$ as a polynomial in $z$. Next, we introduce the polynomial $P$ in $z$ via
\begin{equation}
P(z) = -z G_{n+1,x}(z) - \big(\alpha^2 + \kap\big) F_n(z) - H_n(z) 
\end{equation}
on $\bbC\times\Omega_\mu$.\ Applying once more \eqref{ch3.55} shows that $P$ also has at most degree $n$ in $z$ and hence we write 
$P=P_n$ in the following. One then computes, 
\begin{align}
& G_{n+1}(z) P_n(z) = - (z/2) \partial_x \big[G_{n+1}(z)^2\big] 
- \big(\alpha^2 + \kap\big) F_n(z)G_{n+1}(z) - G_{n+1}(z) H_n(z)    \no \\
& \quad = (z/2) [F_{n,x}(z) H_n(z) + F_{n}(z) H_{n,z}(z)] 
- \big(\alpha^2 + \kap\big) F_n(z)G_{n+1}(z)   \no \\
& \qquad - G_{n+1}(z) H_n(z)    \no \\
& \quad = (z/2) H_{n,x}(z) F_n(z) - \big(\alpha^2 + \kap\big) G_{n+1}(z) F_n(z)   \no \\
& \qquad + H_n(z) [(z/2) F_{n,x}(z) - G_{n+1}(z)]    \no \\
& \quad = \big[(z/2) H_{n,x}(z) - \big(\alpha^2 + \kap\big) G_{n+1}(z) - (\alpha + z)H_n\big] F_n(z). 
\lb{ch3.57}
\end{align}
Temporarily restricting $x\in\wti\Omega_\mu$, where 
\begin{align}
\wti\Omega_\mu&=\{x\in\Omega_\mu\mid  \mu_j(x) F_{n,x}(\mu_j(x),x)/2 = 
- y(\hat\mu_j(x)) = G_{n+1}(\mu_j(x),x) \neq 0, \no \\
& \hspace*{9cm} j=1,\dots,n\} \no \\
&=\{x\in\Omega_\mu\mid
\mu_j(x)\notin\{E_0,\dots,E_{2n+1}\},
\, j=1,\dots,n\},  \lb{ch3.58}
\end{align}
one infers that 
\begin{equation}
P_n(z,x) = \gamma(x) F_n(z,x)    \lb{ch3.59}
\end{equation}
for some continuous function $\gamma$ on $\wti \Omega_{\mu}$. Taking $z=0$ in \eqref{ch3.54} 
then yields 
\begin{equation}
\bigg(\prod_{m=0}^{2n+1} E_m\bigg) - \bigg[\alpha \prod_{j=1}^n \mu_j\bigg]^2 
= (-1)^n \bigg(\prod_{j=1}^n \mu_j \bigg) H_n(0)    \lb{ch3.60}
\end{equation} 
and employing the trace (resp., product) formula for $\kap$ in \eqref{ch3.40}, \eqref{ch3.60} is 
equivalent to  
\begin{equation}
\alpha(x)^2 + \kap(x) = - H_n(0,x)/F_n(0,x), \quad x \in \wti\Omega_\mu.
\end{equation}
Next, choosing $z=0$ in \eqref{ch3.57} implies (with $G_n(0) = \alpha F_n(0)$ by \eqref{ch3.43})
\begin{align}
& 2 G_{n+1}(0,x) P_n(0,x) = 2 \alpha(x) \gamma(x) F_n(0,x)^2   \no \\
& \quad = \big[- 2 \big(\alpha(x)^2 + \kap(x)\big) G_{n+1}(0,x) - 2 \alpha(x) H_n(0,x)\big] F_n(0,x)  
\no \\
& \quad = \{2 [H_n(0,x)/F_n(0,x)] \alpha(x) F_n(0,x) - 2 \alpha(x) H_n(0,x)\} F_n(0,x)   \no \\
& \quad =0, \quad x \in \wti\Omega_\mu.
\end{align}
Thus, 
\begin{equation}
\gamma(x) =0 \, \text{ for $x \in \wti\Omega_\mu$ such that $\alpha(x) \neq 0$.}
\end{equation}
Since $\alpha = (u_x + 2 u) \in C^{\infty}(\Omega_{\mu})$ by hypothesis, and $u(x) \neq e^{- 2x}$, one 
concludes that $\gamma(x) =0$, $x \in \wti\Omega_\mu$. At this point one can follow the final 
part of \cite[Thm.~3.11]{GH03} (or \cite[Ch.~5]{GH03a}) to conclude that 
\begin{equation}
\gamma(x) = 0 \, \text{ and hence $P_n(z,x) = 0$ for $x \in \Omega_{\mu}$.} 
\end{equation}
Thus,
\begin{equation}
H_n(z) = -z G_{n+1}(z) - \big(\alpha^2 + \kap\big) F_n(z)   \lb{ch3.65}
\end{equation}
on $\bbC \times \Omega_{\mu}$.\ Finally, differentiating $R_{2n+2}(z) - G_{n+1}(z,x)^2 = F_n(z,x) H_n(z,x)$ with respect to $x \in \Omega_{\mu}$ and employing 
\eqref{ch3.43} and \eqref{ch3.65} yields 
\begin{align}
F_n H_{n,x} &= - 2 G_{n+1} G_{n+1,x} - F_{n,x} H_n    \no \\
&= F_n\big[- 2 (\alpha + z) G_{n+1,x} + \big(\alpha^2 + \kap\big) F_{n,x}\big]
\end{align}
and hence also 
\begin{equation}
H_{n,x} = - 2 (\alpha + z) G_{n+1,x} + \big(\alpha^2 + \kap\big) F_{n,x}
\end{equation}
on $\bbC \times \Omega_{\mu}$.\ Thus, the zero-curvature equations \eqref{ch2.17}--\eqref{ch2.19} 
have been established on $\Omega_{\mu}$ and one can now follow the discussion in 
Section \ref{s2} to arrive at \eqref{ch3.41}.
\end{proof}

\section{Basic Facts on Self-Adjoint Hamiltonian Systems} \lb{s4}

We now turn to the Weyl--Titchmarsh theory for singular Hamiltonian
(canonical) systems and briefly recall the basic material needed in the
following section. This material is standard and can be found, for
instance, in \cite{CG02}, \cite{HS81}, \cite{HS83}, \cite{HS84},
\cite{KR74}, \cite{LM02}, and the references therein.

\begin{hypothesis} \lb{h4.1}
$(i)$ Define the $2\times 2$ matrix 
$J=\left(\begin{smallmatrix}0& -1 \\ 1 & 0  \end{smallmatrix}\right)$,
and suppose $a_{j,k}, b_{j,k} \in L_{\loc}^1(\bbR)$, $j,k = 1,2$ and  
$A(x)=\big(a_{j,k}(x)\big)_{j,k=1,2}\ge 0$, 
$B(x)=\big(b_{j,k}(x)\big)_{j,k=1,2}=B(x)^*$ for a.e.~$x\in \bbR$.
We consider the Hamiltonian system 
\begin{equation}\lb{HSa}
J \varPsi'(z,x)= [zA(x)+B(x)] \varPsi(z,x), \quad z\in\bbC
\end{equation}
for a.e. $x\in \bbR$, where $z$ plays the role of the spectral
parameter, and where
\begin{equation}\lb{HSb}
\varPsi(z,x) = (\psi_1(z,x)\;\psi_2(z,x))^\top, \quad
\psi_j(z,\dott)\in AC_{\loc}(\bbR), \,\, j=1,2.
\end{equation}
Here $AC_{\loc}(\bbR)$ denotes the set of locally absolutely
continuous functions on $\bbR$ and the $M^*$ and $M^\top$ denote the
adjoint and transpose of a matrix $M$, respectively. \\
$(ii)$ For all nontrivial solutions $\varPsi$ of \eqref{HSa} we assume
the positive definiteness hypothesis $($cf.\ \cite[Sect.\ 9.1]{At64}$)$
\begin{equation}\lb{2.3}
\int_{c}^d dx \, \varPsi(z,x)^*A(x)\varPsi(z,x) > 0\, ,
\end{equation}
on every interval $(c,d)\subset \bbR$, $c<d$. 
\end{hypothesis}

Next, we introduce the vector space ($-\infty\leq a<b\leq\infty$)
\begin{equation}
L^2_A((a,b))=\bigg\{\Phi: (a,b)\to\bbC^2 
\text{ measurable}\,\bigg|\, \int_a^b
dx\,(\Phi(x),A(x)\Phi(x))_{\bbC^2}<\infty \bigg\}, \lb{4.4}
\end{equation}
where $(\Phi,\Psi)_{\bbC^2}=\sum_{j=1}^2 \ol{\phi_j}\psi_j$ denotes the
standard scalar product in $\bbC^2$. Fix a point $x_0\in\bbR$. Then the
Hamiltonian system \eqref{HSa} is said to be in the {\it limit point case}
at $\infty$ (resp., $-\infty$) if for some (and hence for all)
$z\in\bbC\backslash\bbR$, precisely one solution of \eqref{HSa} lies in
$L^2_A((x_0,\infty))$ (resp., $L^2_A((-\infty,x_0))$). (By the analog
of Weyl's alternative, if \eqref{HSa} is not in the limit point case at
$\pm\infty$,  all solutions of \eqref{HSa} lie in $L^2_A((x_0,\pm\infty))$
for all $z\in\bbC$. In the latter case the Hamiltonian system
\eqref{HSa} is said to be in the {\it limit circle case} at $\pm\infty$.)

To simplify matters for the remainder of this section, we will
always suppose the limit point case at $\pm\infty$ from now on.

\begin{hypothesis} \lb{h4.2}
Assume Hypothesis \ref{h4.1} and suppose that the Hamiltonian system
\eqref{HSa} is in the limit point case at $\pm\infty$.
\end{hypothesis}

An elementary example of a Hamiltonian system satisfying Hypothesis
\ref{h4.2} is given by the case where all entries of $A$ and $B$ are
essentially bounded on $\bbR$ (cf.\ Section \ref{s5}).

When considering the Hamiltonian system \eqref{HSa} on the half-line
$[x_0,\infty)$ (resp., $(-\infty,x_0]$), a self-adjoint (separated)
boundary condition at the point $x_0$ is of the type
\begin{equation}
\beta\Psi(x_0)=0, \lb{4.5}
\end{equation}
where $\beta=(\beta_1\;\beta_2)\in\bbC^{1\times 2}$ satisfies
\begin{equation}
\beta\beta^*=1, \quad \beta J\beta^*=0 \; 
\text{ (equivalently, $|\beta_1|^2+|\beta_2|^2=1$, 
$\Im(\beta_2 \ol{\beta_1})=0$).} \lb{4.6}
\end{equation}
In particular, the boundary condition \eqref{4.5} (with $\beta$
satisfying \eqref{4.6}) is equivalent to
$\beta_1\psi_1(x_0)+\beta_2\psi_2(x_0)=0$ with
$\beta_1/\beta_2\in\bbR$ if $\beta_2\neq 0$ and 
$\beta_2/\beta_1\in\bbR$ if $\beta_1\neq 0$. The special case
$\beta_0=(1\; 0)$ will be of particular relevance in Section \ref{s5}. Due
to our limit point assumption at $\pm\infty$ in Hypothesis
\ref{h4.2}, no additional boundary condition at $\pm\infty$ needs to
be introduced when considering \eqref{HSa} on the half-lines
$[x_0,\infty)$ and $(-\infty,x_0]$. The resulting full-line and
half-line Hamiltonian systems are said to be self-adjoint on $\bbR$,
$[x_0,\infty)$, and $(-\infty,x_0]$, respectively (assuming of course a
boundary condition of the type \eqref{4.5} in the two half-line cases).

Next we digress a bit and briefly turn to Nevanlinna--Herglotz functions and their
representations in terms of measures, the focal point of
Weyl--Titchmarsh theory (and hence spectral theory) of self-adjoint
Hamiltonian systems.

\begin{definition} \lb{d4.3}
Any analytic map $m\colon\bbC_+\to\bbC$ is called a {\it Nevanlinna--Herglotz}
function if 
$\Im(m(z))\geq 0$ for all $z\in\bbC_+$ $($here $\bbC_+=\{z\in\bbC\,|\,\Im(z)>0\}$$)$. Similarly, any
analytic map $M\colon\bbC_+\to\bbC^{k\times k}$,
$k\in\bbN$, is called a $k\times k$ matrix-valued Nevanlinna--Herglotz function if 
$\Im(M(z))\geq 0$ for all $z\in\bbC_+$.
\end{definition}

Nevanlinna--Herglotz functions are characterized by a representation of the form 
\begin{align}
&m(z)= a + bz  +\int_{-\infty}^\infty
d\omega(\lambda)\,\big((\lambda-z)^{-1} -\lambda(1+\lambda^2)^{-1}\big),
\quad z\in\bbC\backslash\bbR, \lb{HF} \\ 
&a \in\bbR, \; b\geq 0, \quad \int_{-\infty}^{\infty} 
d\omega(\lambda)\,(1+\lambda^2)^{-1}  < \infty, \lb{HFa} \\
&\omega((\lambda_1, \lambda_2])= \lim_{\delta\downarrow 0}
\lim_{\epsilon \downarrow 0}\frac{1}{\pi}\int_{\lambda_1
+ \delta}^{\lambda_2 + \delta }d\nu\,  \Im\left( m(\nu
+i\epsilon)\right), \lb{HFb}
\end{align}
in the following sense: Every Nevanlinna--Herglotz function admits a representation
of the type \eqref{HF}, \eqref{HFa} and conversely, any function of the
type \eqref{HF}, \eqref{HFa} is a Nevanlinna--Herglotz function. Moreover, local 
singularities and zeros of $m$ are necessarily located on the real axis and
at most of first order in the sense that
\begin{align}
&\omega(\{\lambda\})=\lim_{\varepsilon\downarrow0}\left(\omega(\lambda
+\varepsilon)-\omega(\lambda-\varepsilon)\right)=
-\lim_{\varepsilon\downarrow0} i\varepsilon\,
m(\lambda+i\varepsilon) \geq 0, \quad \lambda\in\bbR, \lb{4.10} \\
&\lim_{\varepsilon\downarrow0} i\varepsilon \, m(\lambda
+i\varepsilon)^{-1}\geq 0, \quad \lambda\in\bbR. \lb{4.11}
\end{align}
In particular, isolated poles of $m$ are simple and located
on the real axis, the corresponding residues being negative. Analogous
results hold for matrix-valued Nevanlinna--Herglotz functions (see, e.g., \cite{GT97}
and the literature cited therein).

For subsequent purpose in Section \ref{s5} we also note that
$- z^{-1}$ is a Nevanlinna--Herglotz function and that compositions of Nevanlinna--Herglotz 
functions remain Nevanlinna--Herglotz functions (as long as this composition is well-defined). In addition, diagonal elements of a matrix-valued
Nevanlinna--Herglotz function are Nevanlinna--Herglotz functions. 

Returning to Hamiltonian systems on half-lines $[x_0,\pm\infty)$
satisfying Hypotheses \ref{h4.1} and \ref{h4.2}, we now denote by
$\Psi_\pm(z,x,x_0)$ the unique solution of \eqref{HSa} satisfying
$\Psi_\pm(z,\dott,x_0)\in L^2_A([x_0,\pm\infty))$,
$z\in\bbC\backslash\bbR$, normalized by $\psi_{1,\pm}(z,x_0,x_0)=1$. Then
the half-line Weyl--Titchmarsh function $m_\pm(z,x)$, associated with the
Hamiltonian system \eqref{HSa} on $[x,\pm\infty)$ and the fixed boundary
condition $\beta_0=(1\;0)$ at the point $x\in\bbR$, is defined by
\begin{equation}
m_\pm(z,x)=\psi_{2,\pm}(z,x,x_0)/\psi_{1,\pm}(z,x,x_0), \quad
z\in\bbC\backslash\bbR, \; \pm x\geq \pm x_0.  \lb{4.12}
\end{equation}
The actual normalization of $\Psi_\pm(z,x,x_0)$ was chosen for
convenience only and is clearly immaterial in the definition of
$m_\pm(z,x)$ in \eqref{4.12}. For later use in Section \ref{s5} we also recall that 
\begin{equation}
\Psi_\pm(z,x,x_0) = \begin{pmatrix}\psi_{1,\pm}(z,x,x_0) \\
\psi_{2,\pm}(z,x,x_0)  \end{pmatrix} = \begin{pmatrix}\vartheta_1(z,x,x_0)
& \varphi_1(z,x,x_0) \\ \vartheta_2(z,x,x_0)
& \varphi_2(z,x,x_0)\end{pmatrix}
\begin{pmatrix} 1 \\  m_\pm(z,x_0) \end{pmatrix},     \lb{4.12a}
\end{equation}
with $\vartheta_j(z,x,x_0)$, and $\varphi_j(z,x,x_0)$, $j=1,2$, defined such that 
\begin{equation}
\Upsilon(z,x,x_0) = \begin{pmatrix} \vartheta_1(z,x,x_0) & \varphi_1(z,x,x_0)\\
\vartheta_2(z,x,x_0) & \varphi_2(z,x,x_0)
\end{pmatrix}
\end{equation}
represents a normalized fundamental system of solutions of \eqref{HSa} at some
$x_0\in\bbR$, satisfying 
\begin{equation} 
\Upsilon(z,x_0,x_0) = I_2.
\end{equation}
One recalls that for fixed $x, x_0 \in \bbR$, 
\begin{equation}
\text{$\vartheta_j(z,x,x_0)$ and $\varphi_j(z,x,x_0)$, 
$j=1,2$, are entire in $z\in\bbC$.}      \lb{4.12b}
\end{equation}

In addition, one verifies that
$m_\pm(z,x)$ satisfies the following Riccati-type differential equation, 
\begin{align}
&m'(z,x) +[b_{2,2}(x)+a_{2,2}(x)z]m(z,x)^2 \lb{4.13} \\
&\quad + [b_{1,2}(x)+b_{2,1}(x)+(a_{1,2}(x)+a_{2,1}(x))z]m(z,x)  
+ b_{1,1}(x) +a_{1,1}(x)z =0. \no 
\end{align}

Finally, the $2\times 2$ Weyl--Titchmarsh matrix $M(z,x)$ associated
with the Hamiltonian system  \eqref{HSa} on $\bbR$ is then defined in
terms of the half-line Weyl--Titchmarsh functions $m_\pm(z,x)$ by
\begin{align}
M(z,x)
&=\big(M_{j,j^\prime}(z,x)\big)_{j,j^\prime=1,2},  \quad 
z\in\bbC\backslash\bbR,  \lb{4.14} \\
M_{1,1}(z,x)&=[m_-(z,x)-m_+(z,x)]^{-1}, \no \\ 
M_{1,2}(z,x)&=M_{2,1}(z,x) \no  \\
&=2^{-1}[m_-(z,x)-m_+(z,x)]^{-1}
[m_-(z,x)+m_+(z,x)], \lb{4.15} \\
M_{2,2}(z,x)&=[m_-(z,x)-m_+(z,x)]^{-1}
m_-(z,x)m_+(z,x). \no 
\end{align}
One verifies that $M(z,x)$ is a $2\times 2$ matrix-valued Nevanlinna--Herglotz 
function. We emphasize that for any fixed $x_0\in\bbR$, $M(z,x_0)$
contains all the spectral information of the self-adjoint Hamiltonian
system \eqref{HSa} on $\bbR$ (assuming Hypotheses \ref{h4.1} and
\ref{h4.2}).

\section{Real-Valued Algebro-Geometric \texorpdfstring{$\CH$}{CH-2} Solutions and the Associated
Isospectral Torus} \lb{s5}

In our final and principal section we study real-valued algebro-geometric
solutions of the $\CH$ hierarchy associated with curves $\calK_n$ whose
affine part is nonsingular and prove that the isospectral manifold of
smooth bounded solutions of the $n$th stationary $\CH$ equation can be characterized as 
a real $n$-dimensional torus $\bbT^n$. We focus on the stationary case  as
this is the primary concern in this context and briefly comment on the time-dependent case 
at the end of this section.  

To study the direct spectral problem we first introduce the following
assumptions.

\begin{hypothesis} \lb{h5.1}
Suppose 
\begin{equation}
E_0 < E_1 < \cdots < E_{2n} < E_{2n+1}, \quad 0 \in (E_{2m_0}, E_{2m_0 + 1})     \lb{5.1}
\end{equation}
for some $m_0 \in \{0,\dots,n\}$, and let $u, \kap$ be a real-valued solution 
of the $n$th stationary $\CH$ equation \eqref{ch2.29a},
\begin{equation}
\sCH_{n}(u, \kap)=0   \lb{5.2}
\end{equation} 
satisfying 
\begin{equation}
u, \kap \in C^\infty(\bbR), \; \kap > 0, \; u^{(m)}, \kap^{(m)} \in L^\infty(\bbR), \; 
m\in\bbN_{0}.    \label{5.3} 
\end{equation}
\end{hypothesis}

We start by noticing that the basic stationary equation \eqref{ch3.22},
\begin{equation}
\Psi_x(-z,x)=U(-z,x)\Psi(-z,x), \quad \Psi=(\psi_1, \psi_2)^\top, \;
(z,x)\in\bbC\times\bbR, \lb{5.5} 
\end{equation}
is equivalent to the following Hamiltonian (canonical) system
\begin{equation} 
J \wti\Psi_x(\ti z,x)=[\ti z A(x)+B(x)]\wti \Psi(\ti z,x), \quad 
\wti\Psi= \big(\wti\psi_1,\wti\psi_2\big)^\top, \;
\big(\ti z = - z^{-1},x\big)\in\bbC\times\bbR, \lb{5.6} 
\end{equation}
where
\begin{align}
& J=\begin{pmatrix} 0 & -1 \\ 1 & 0 \end{pmatrix}, \quad 
\wti\Psi(\ti z,x)=\Psi(-z,x), \quad \ti z= - z^{-1}, \lb{5.7} \\
&A(x)=\begin{pmatrix} \alpha(x)^2 + \kap(x) & - \alpha(x) \\ 
- \alpha(x) & 1 \end{pmatrix} > 0, \quad 
B(x)=\begin{pmatrix} 0 & -1 \\ -1 & 0 \end{pmatrix} = B(x)^*, \quad x\in\bbR.
\lb{5.8} 
\end{align}
In particular, due to assumptions \eqref{5.1}--\eqref{5.3}, the
Hamiltonian system \eqref{5.6} satisfies Hypotheses \ref{h4.1} and
\ref{h4.2}. Explicitly, the Hamiltonian system
\eqref{5.6} boils down to
\begin{align}
\wti\psi_{1,x}(\ti z,x)&= - \wti\psi_1(\ti z,x) - \tilde z \alpha(x) \wti \psi_1(\ti z,x)
+ \tilde z \wti\psi_2(\ti z,x), \lb{5.8a} \\
\wti\psi_{2,x}(\ti z,x)&= \wti\psi_2(\ti z,x) + \tilde z \alpha(x) \wti\psi_2(\ti z,x) 
- \ti z \big(\alpha(x)^2 + \kap(x)\big)\wti\psi_1(\ti z,x),    \lb{5.8b} \\ 
& \hspace*{4.9cm} \big(z=-\ti z^{-1},x\big)\in\bbC\times\bbR,  \no
\end{align}
and upon eliminating $\wti\psi_2$ results in a particular case of a quadratic 
weighted Sturm--Liouville pencil (cf. \cite{CLZ06}, \cite{CI08}, \cite{Ec15}--\cite{EK15}, \cite{HI11}) of the type 
\begin{align}
\begin{split} 
-\wti\psi_{1,xx}(\wti z,x)+\wti\psi_1(\ti z,x) = {\ti z}^2 \kap(x) \wti\psi_1(\ti z,x) 
- \ti z(4u(x)-u_{xx}(x))\wti\psi_1(\ti z,x),& \\
\big(z = - \ti z^{-1},x\big)\in\bbC\times\bbR.&     \lb{5.8d}
\end{split} 
\end{align}
Introducing 
\begin{equation}
\Sigma_n=\bigcup_{\ell=0}^{n} [E_{2\ell},E_{2\ell+1}], \lb{5.9}
\end{equation} 
we define
\begin{align}
&R_{2n+2} (\lambda)^{1/2} = |R_{2n+2} (\lambda)^{1/2}| \no \\
&\times \begin{cases}
-1
& \text{for $\lambda \in (E_{2n+1}, \infty)$}, \\
(-1)^{n+j}
& \text{for $\lambda \in (E_{2j-1}, E_{2j})$,
$j=1,\dots,n$},  \\
(-1)^n
& \text{for $\lambda \in (-\infty, E_0)$},  \\
i(-1)^{n+j+1}  
& \text{for $\lambda \in (E_{2j}, E_{2j+1})$, 
$j=0,\dots,n$}, \end{cases} \qquad \lambda\in\bbR, \lb{5.11} 
\end{align}
and
\begin{equation}
{R}_{2n+2} (\lambda)^{1/2}
=\lim\limits_{\varepsilon \downarrow 0} {R}_{2n+2} 
(\lambda +i\varepsilon)^{1/2}, \quad \lambda \in \Sigma_n,
\lb{5.10}
\end{equation}
and analytically continue $R_{2n+2} (\dott)^{1/2}$ to
$\bbC\backslash\Sigma_n$. We also note the property
\begin{equation}
\ol{R_{2n+2}(\ol z)^{1/2}}=R_{2n+2}(z)^{1/2}. \lb{5.11a}
\end{equation}
For notational convenience we will occasionally call $(E_{2j-1},E_{2j})$,
$j=1,\dots,n$, spectral gaps and $E_{2j-1}, E_{2j}$ the corresponding
spectral gap endpoints.

Next, we introduce the cut plane
\begin{equation}
\Pi_n = \bbC\backslash\Sigma_n, \lb{5.12}
\end{equation}
and the upper, respectively, lower sheets $\Pi_{n,\pm}$ of $\calK_n$ by
\begin{equation}
\Pi_{n,\pm} = \{ (z, \pm R_{2n+2} (z)^{1/2})
\in \calK_n\mid z\in\Pi_n\},     \lb{5.13}
\end{equation}
with the associated charts
\begin{equation}
\zeta_\pm \colon \Pi_{n,\pm} \to \Pi_n,\quad P=(z,\pm R_{2n+2}(z)^{1/2}) 
\mapsto z. \lb{5.14}
\end{equation}
The two branches $\Psi_{\pm}(z,x,x_0)$ of the
Baker--Akhiezer vector $\Psi(P,x,x_0)$ in \eqref{ch3.15} are then given by
\begin{equation}
\Psi_\pm (z,x,x_0)=\Psi(P,x,x_0), \quad P=(z,y)\in\Pi_{n,\pm}, \quad 
\Psi_\pm =(\psi_{1,\pm},\psi_{2,\pm})^\top, \lb{5.15}
\end{equation}
and one infers from \eqref{ch3.38c} (note that the error term is uniform in $x$ in the case where $\mu_j(x)$ remains within its respective gaps) that 
\begin{equation}
\psi_{1,\pm} (z,\dott,x_0)\in L^2((x_0,\mp\infty)) \, \text{ for $|z|$
sufficiently large.} \lb{5.16} 
\end{equation}
Thus, introducing 
\begin{equation}
\wti\Psi_\pm (\ti z,x,x_0)=\Psi_\mp(-z,x,x_0), \quad 
\wti\Psi_\pm = \big(\wti\psi_{1,\pm},\wti\psi_{2,\pm}\big)^\top, \quad \ti
z= - z^{-1}, \lb{5.17}
\end{equation}
and the two branches $\phi_\pm(z,x)$ of $\phi(P,x)$ on $\Pi_{n,\pm}$ by
\begin{equation}
\phi_\pm (z,x)=\phi(P,x), \quad P=(z,y)\in\Pi_{n,\pm}, \lb{5.18}
\end{equation}
one infers from \eqref{4.12} and
\eqref{5.16} that the Weyl--Titchmarsh functions $\wti m_\pm(\ti z,x)$
associated with the self-adjoint Hamiltonian system
\eqref{5.6} on the half-lines $[x,\pm\infty)$ and the Dirichlet boundary
condition indexed by $\beta_0=(1\;0)$ at the point $x\in\bbR$, are given
by
\begin{align}
\wti m_\pm(\ti z,x)&=\wti\psi_{2,\pm}(\wti z,x,x_0)
/\wti\psi_{1,\pm}(\ti z,x,x_0)=
\psi_{2,\mp}(-z,x,x_0)/\psi_{1,\mp}(-z,x,x_0) \no \\
&= - \phi_{\mp}(-z,x), \quad z\in\bbC\backslash\Sigma_n. \lb{5.19} 
\end{align}
More precisely, \eqref{5.16} yields \eqref{5.19} only for $|z|$
sufficiently large. However, since by general principles 
$\wti m_\pm(\dott,x)$ are analytic in $\bbC\backslash\bbR$, and by
\eqref{ch3.11}, $\phi_\pm(\dott,x)$ are analytic in $\bbC\backslash\Sigma_n$,
one infers \eqref{5.19}. An application of \eqref{4.12a} and \eqref{4.12b} then shows that 
\eqref{5.16} extends to all $z\in\bbC\backslash\Sigma_n$, that is,
\begin{equation}
\psi_{1,\pm} (z,\dott,x_0)\in L^2((x_0,\mp\infty)), \quad 
z\in\bbC\backslash\Sigma_n. \lb{5.19a} 
\end{equation}

Next, we mention a useful fact concerning a special class of Nevanlinna--Herglotz 
functions closely related to the problem at hand. The result must be
well-known to experts, but since we could not quickly locate a proof in the
literature, we provide the simple contour integration argument below.

\begin{lemma} \lb{l5.2}
Let $P_N$ be a monic polynomial of degree $N$. Then 
$P_N/R_{2n+2}^{1/2}$, respectively, $- P_N/R_{2n+2}^{1/2}$ is a Nevanlinna--Herglotz 
function if and only if one of the following three cases applies: \\
$(i)$ $N=n$ and
\begin{equation}
P_n(z)=\prod_{j=1}^n (z-a_j), \quad a_j\in [E_{2j-1},E_{2j}], \;
j=1,\dots,n. \lb{5.22}
\end{equation}
If \eqref{5.22} is satisfied, then $P_n/R_{2n+2}^{1/2}$ admits the
Nevanlinna--Herglotz representation
\begin{equation}
\f{P_n(z)}{R_{2n+2}(z)^{1/2}}=\f{1}{\pi}\int_{\Sigma_n} 
\f{|P_n(\lambda)|\,d\lambda}{|R_{2n+2}(\lambda)^{1/2}|}\f{1}{\lambda-z}, 
\quad z\in\bbC\backslash\Sigma_n. \lb{5.23}
\end{equation}
$(ii)$ $N=n+1$ and 
\begin{equation}
P_{n+1}(z)=\prod_{\ell=0}^n (z-b_\ell), \quad 
b_0 \in (-\infty, E_0], \;\, b_j\in [E_{2j-1},E_{2j}], \;
j=1,\dots,n. \lb{5.24}
\end{equation}
If \eqref{5.24} is satisfied, then $P_{n+1}/R_{2n+2}^{1/2}$ admits
the Nevanlinna--Herglotz representation
\begin{align}
&\f{P_{n+1}(z)}{R_{2n+2}(z)^{1/2}} =
\Re\bigg(\f{P_{n+1}(i)}{R_{2n+2}(i)^{1/2}}\bigg)+\f{1}{\pi}\int_{\Sigma_n} 
\f{|P_{n+1}(\lambda)|\,d\lambda}{|R_{2n+2}(\lambda)^{1/2}|}
\bigg(\f{1}{\lambda-z}-\f{\lambda}{1+\lambda^2}\bigg), \no \\
&\hspace*{9.5cm} z\in\bbC\backslash\Sigma_n. \lb{5.25} 
\end{align}
$(iii)$ $N=n+1$ and 
\begin{equation}
P_{n+1}(z)=\prod_{\ell=0}^n (z-d_\ell), \quad 
d_0 \in [E_{2n+1},\infty), \;\, d_j\in [E_{2j-1},E_{2j}], \;
j=1,\dots,n. \lb{5.26}
\end{equation}
If \eqref{5.26} is satisfied, then $- P_{n+1}/R_{2n+2}^{1/2}$ admits
the Nevanlinna--Herglotz representation
\begin{align}
&\f{- P_{n+1}(z)}{R_{2n+2}(z)^{1/2}} =
- \Re\bigg(\f{P_{n+1}(i)}{R_{2n+2}(i)^{1/2}}\bigg)+\f{1}{\pi}\int_{\Sigma_n} 
\f{\big|P_{n+1}(\lambda)\big|\,d\lambda}{|R_{2n+2}(\lambda)^{1/2}|}
\bigg(\f{1}{\lambda-z}-\f{\lambda}{1+\lambda^2}\bigg), \no \\
&\hspace*{9.5cm} z\in\bbC\backslash\Sigma_n. \lb{5.26a} 
\end{align}
\end{lemma}
\begin{proof} 
Except for case $(iii)$ this has been proven in \cite[Lemma~5.1]{GH08}. 
For convenience of the reader we repeat the argument here. 
Since Nevanlinna--Herglotz functions are $\Oh(z)$ as $|z|\to\infty$ and
cannot vanish faster than $\Oh\big(z^{-1}\big)$ as $|z|\to\infty$, we can confine
ourselves to the range $N\in\{n,n+1,n+2\}$. We start with the case $N=n$
and employ the following contour integration approach. Consider a closed 
oriented contour $\Gamma_{R,\varepsilon}$ which consists of the
clockwise oriented semicircle
$C_{\varepsilon}=\{z\in\bbC\,|\, z=E_0-\varepsilon \exp(-i\theta), \,
-\pi/2\leq\theta\leq \pi/2\}$ centered at $E_0$, the straight line
$L_+=\{z\in\bbC_+\,|\, z=E_0+x+i\varepsilon, \, 0\leq x \leq R\}$
(oriented from left to right), the following part of the counterclockwise
oriented circle of radius $(R^2+\varepsilon^2)^{1/2}$ centered at $E_0$,
$C_R=\{z\in\bbC\,|\, z=E_0+(R^2+\varepsilon^2)^{1/2}\exp(i\theta), \,
\arctan(\varepsilon/R) \leq \theta \leq 2\pi-\arctan(\varepsilon/R)\}$,
and the straight line $L_-=\{z\in\bbC_-\,|\, z=E_0+x-i\varepsilon, \,
0\leq x\leq R\}$ (oriented from right to left). Then, for
$\varepsilon>0$ small enough and
$R>0$ sufficiently large, one infers
\begin{align}
\f{P_n(z)}{R_{2n+2}(z)^{1/2}}&=\f{1}{2\pi
i}\oint_{\Gamma_{R,\varepsilon}}
\f{1}{\zeta-z}\f{P_n(\zeta)}{R_{2n+2}(\zeta)^{1/2}}d\zeta \no \\
&\hspace*{-.45cm} \underset{\varepsilon\downarrow 0, R\uparrow\infty}{=}
\f{1}{\pi} \int_{\Sigma_n} 
\f{1}{\lambda-z}\f{P_n(\lambda)d\lambda}{iR_{2n+2}(\lambda)^{1/2}}.
\lb{5.27}
\end{align} 
Here we used \eqref{5.11} to compute the contributions of the contour
integral along $[E_0,R]$ in the limit $\varepsilon\downarrow 0$ and note
that the integral over $C_R$ tends to zero as $R\uparrow \infty$ since 
\begin{equation}
\f{P_n(\zeta)}{R_{2n+2}(\zeta)^{1/2}}\underset{\zeta\to\infty}{=}
\Oh\big(|\zeta|^{-1}\big). \lb{5.28}
\end{equation}
Next, utilizing the fact that $P_n$ is monic and using \eqref{5.11} again,
one infers that $P_n(\lambda)d\lambda/[iR_{2n+2}(\lambda)^{1/2}]$
represents a positive measure supported on $\Sigma_n$ if and only if $P_n$
has precisely one zero in each of the intervals $[E_{2j-1},E_{2j}]$, 
$j=1,\dots,n$. In other words, 
\begin{equation}
\f{P_n(\lambda)}{iR_{2n+2}(\lambda)^{1/2}}=
\f{|P_n(\lambda)|}{|R_{2n+2}(\lambda)^{1/2}|}\geq 0 \,\text{ on } \Sigma_n 
\lb{5.29}
\end{equation}
if and only if $P_n$ has precisely one zero in each of the intervals 
$[E_{2j-1},E_{2j}]$, $j=1,\dots,n$. The Nevanlinna--Herglotz representation
\eqref{HF}, \eqref{HFa} then finishes the proof of \eqref{5.23}. 

In the case where $N=n+1$, the proofs of \eqref{5.24} and \eqref{5.26} follow along similar 
lines taking into account the additional residues at $\pm i$ inside
$\Gamma_{R,\varepsilon}$ which are responsible for the real parts on the
right-hand sides of \eqref{5.25} and \eqref{5.26a}.

Finally, in the case $N=n+2$, assume that $P_{n+2}/R_{2n+2}^{1/2}$ is a
Nevanlinna--Herglotz function. Then for some $a\in\bbR$, $b\geq 0$, and 
some finite (positive) measure $\omega$ supported on $[E_0,E_{2n+2}]$, 
\begin{equation}
\f{P_{n+2}(z)}{R_{2n+2}(z)^{1/2}}=a+bz+\int_{E_0}^{E_{2n+2}} d\omega(\lambda)\, 
(\lambda-z)^{-1}, \quad 
z\in\bbC\backslash\Sigma_n,  \lb{5.30}
\end{equation}
since
\begin{equation}
\lim_{\varepsilon\downarrow 0}
\Im(P_{n+2}(\lambda)R_{2n+2}(\lambda+i\varepsilon)^{-1/2})=0 \, 
\text{ for $\lambda>E_{2n+2}$ and $\lambda<E_0$.} \lb{5.31}
\end{equation}
In particular, \eqref{5.30} implies
\begin{equation}
P_{n+2}(z)R_{2n+2}(z)^{-1/2}\underset{|z|\to\infty}{=}bz+\Oh(1), \quad
b\geq 0. \lb{5.32}
\end{equation}
However, by \eqref{5.11}, one immediately infers
\begin{equation}
P_{n+2}(\lambda)R_{2n+2}(\lambda)^{-1/2}
\underset{\lambda\uparrow\infty}{=}-\lambda+\Oh(1). \lb{5.33}
\end{equation}
This contradiction dispenses with the case $N=n+2$.
\end{proof}

Now we are in position to state the following result concerning the
half-line and full-line Weyl--Titchmarsh functions associated with the
self-adjoint Hamiltonian system \eqref{5.6}. We denote by $\wti m_\pm(\ti
z,x)$ the Weyl--Titchmarsh $m$-functions corresponding to \eqref{5.6}
associated with the half-lines $(x,\pm\infty)$ and the Dirichlet boundary
condition indexed by $\beta_0=(1\;0)$ at the point $x\in\bbR$, and by
$\wti M(\ti z,x)$ the $2\times 2$ Weyl--Titchmarsh matrix corresponding to
\eqref{5.6} on $\bbR$ (cf.\ \eqref{4.12}, \eqref{4.14}, and \eqref{4.15}).
Moreover, $\Sigma_n^o$ denotes the open interior of $\Sigma_n$ and the real
part of a matrix $M$ is defined as usual by $\Re(M)=(M+M^*)/2$.

\begin{theorem} \lb{t5.3}
Assume Hypothesis \ref{h5.1} and let $(z,x)\in(\bbC\backslash\Sigma_n)\times\bbR
$, $\ti z= - z^{-1}$. Then
\begin{align}
\wti m_\pm(\ti z,x)&= \f{\pm R_{2n+2}(-z)^{1/2} + G_{n+1}(-z,x)}{F_n(-z,x)}
\lb{5.20} \\
&= \pm z - z + c_1 + \sum_{j=1}^n \mu_j(x) 
\pm \Re\bigg(\f{R_{2n+2}(- i)^{1/2}}{iF_n(- i,x)}\bigg)    \no \\ 
& \quad - \sum_{j=1}^n \f{G_{n+1}(- \mu_j(x),x) \big[1 \mp \wti \varepsilon_j (x)\big]}
{F_{n,z}(- \mu_j(x),x)}\f{1}{z + \mu_j(x)}    \lb{5.20a} \\
& \quad \pm \f{1}{\pi} \int_{\Sigma_n} 
\f{|R_{2n+2}(- \lambda)^{1/2}|\, d\lambda}{|F_n(- \lambda,x)|}
\bigg(\f{1}{\lambda -z}-\f{\lambda}{1+\lambda^2}\bigg),     \no  
\end{align}
where $\wti \varepsilon_j(x)\in\{1,-1\}$, $j=1,\dots,n$, is chosen such that 
\begin{equation}
\f{G_{n+1}(- \mu_j(x),x)\wti \varepsilon_j(x)}{F_{n,z}(- \mu_j(x),x)}\geq 0, \quad
j=1,\dots,n. \lb{5.20b}
\end{equation}
Moreover,
\begin{align}
\wti M(\ti z,x)&= \f{- 1}{2R_{2n+2}(-z)^{1/2}} \begin{pmatrix}
F_n(-z,x) & G_{n+1}(-z,x) \\ G_{n+1}(-z,x) & - H_n(-z,x) \end{pmatrix} \lb{5.21} \\
&= \Re\big(\wti M(-i,x)\big)+\int_{\Sigma_n} d\Omega(\lambda,x)\,
\bigg(\f{1}{\lambda-z}-\f{\lambda}{1+\lambda^2}\bigg), \lb{5.21A}
\end{align}
where
\begin{equation}
\Omega(\lambda,x)=\f{- 1}{2\pi i R_{2n+2}(-\lambda)^{1/2}}\begin{pmatrix}
F_n(- \lambda,x) & G_{n+1}(- \lambda,x) \\ 
G_{n+1}(- \lambda,x) & - H_n(- \lambda,x) \end{pmatrix}, \quad
\lambda \in\Sigma_n^o. \lb{5.21B}
\end{equation}  
The essential spectrum of the half-line Hamiltonian systems
\eqref{5.6} on $[x,\pm\infty)$ $($with any self-adjoint boundary condition 
at $x$$)$ as well as the essential spectrum of the Hamiltonian system
\eqref{5.6} on $\bbR$ is purely absolutely continuous and given by
\begin{equation}
\big(-\infty, - E_{2m_0 + 1}^{-1}\big] \cup \bigcup_{\substack{\ell=0 \\
\ell \neq m_0}}^{n-1} \big[-E_{2l}^{-1},-E_{2\ell+1}^{-1}\big]\cup
\big[-E_{2m_0}^{-1},\infty\big).  \lb{5.21a}
\end{equation}
The spectral multiplicities are simple in the half-line cases and of
uniform multiplicity two in the full-line case.
\end{theorem}
\begin{proof}
Equation \eqref{5.20} follows from \eqref{ch3.11}, \eqref{5.11}, and
\eqref{5.19}. Equation \eqref{5.21} is then a consequence of
\eqref{ch3.19}--\eqref{ch3.21}, \eqref{4.14}, \eqref{4.15},
\eqref{5.19}, and \eqref{5.20}. Different self-adjoint boundary conditions
at the point $x$ lead to different half-line Hamiltonian systems whose
Weyl--Titchmarsh functions are related by a linear fractional
transformation (cf., e.g., \cite{CG02}), which leads to the invariance of
the essential spectrum with respect to the boundary condition at $x$. In
order to prove the Nevanlinna--Herglotz representation \eqref{5.20a} one can follow
the corresponding computation for Schr\"odinger operators with
algebro-geometric potentials in \cite[Sect.\ 8.1]{Le87}. For this purpose
one first notes that by \eqref{5.25} also
$- R_{2n+2}(z)^{1/2}/F_n(z,x)$ is a Nevanlinna--Herglotz function. A contour
integration as in the proof of Lemma
\ref{l5.2} then proves
\begin{align}
\f{R_{2n+2}(z)^{1/2}}{F_n(z,x)} 
&= z + \Re\bigg(\f{R_{2n+2}(i)^{1/2}}{F_n(i,x)}\bigg) 
- \sum_{j=1}^n \f{|R_{2n+2}(\mu_j(x))^{1/2}|}
{|F_{n,z}(\mu_j(x),x)|}\f{1}{z-\mu_j(x)} \no \\
& \quad + \f{1}{\pi}\int_{\Sigma_n} 
\f{|R_{2n+2}(\lambda)^{1/2}|\, d\lambda}{|F_n(\lambda,x)|}
\bigg(\f{1}{\lambda -z}-\f{\lambda}{1+\lambda^2}\bigg)  \lb{5.21C} \\
&= z + \Re\bigg(\f{R_{2n+2}(i)^{1/2}}{F_n(i,x)}\bigg) - 
\sum_{j=1}^n \f{G_{n+1}(\mu_j(x),x)\varepsilon_j(x)}
{F_{n,z}(\mu_j(x),x)}\f{1}{z-\mu_j(x)} \no \\
& \quad + \f{1}{\pi}\int_{\Sigma_n} 
\f{|R_{2n+2}(\lambda)^{1/2}|\, d\lambda}{|F_n(\lambda,x)|}
\bigg(\f{1}{\lambda -z}-\f{\lambda}{1+\lambda^2}\bigg). \lb{5.21D}
\end{align}
The only difference compared to the corresponding argument in the proof of
Lemma \ref{l5.2} concerns additional (approximate) semicircles of
radius $\varepsilon$ centered at each $\mu_j(x)$, $j=1,\dots,n$, in the
upper and lower complex half-planes. Whenever, $\mu_j(x)\in
(E_{2j-1},E_{2j})$, the limit $\varepsilon\downarrow 0$ picks up a residue
contribution displayed in the sum over $j$ in \eqref{5.21C}. This
contribution vanishes, however, if $\mu_j(x)\in\{E_{2j-1}, E_{2j}\}$. In
this case $F_{n,z}(\mu_j(x),x) \neq 0$ by \eqref{4.10} and
$R_{2n+2}(\mu_j(x))=0$ by \eqref{ch2.22}. Equation \eqref{5.21D} then
follows from \eqref{ch3.6} and the sign of $\varepsilon_j(x)$ must be
chosen according to
\begin{equation}
\f{G_{n+1}( \mu_j(x),x)\varepsilon_j(x)}{F_{n,z}(\mu_j(x),x)}\geq 0, \quad
j=1,\dots,n,    \lb{5.21DE}
\end{equation}
in order to guarantee nonpositive
residues in \eqref{5.21D} (cf.\ \eqref{4.10}). 

In order to analyze the term $G_{n+1}/F_n$ in $\wti m_{\pm}$ we turn to 
Lagrange-type interpolation formulas. If $Q_{n-1}$ is a
polynomial of degree $n-1$, then
\begin{equation}
Q_{n-1}(z)=F_n(z)\sum_{j=1}^n
\f{Q_{n-1}(\mu_j)}{F_{n,z}(\mu_j)}\f{1}{z-\mu_j}, \quad z\in\bbC.
\lb{5.21E}
\end{equation}
Since $F_n$ and $G_{n+1}$ are monic polynomials of degree $n$ and $n+1$, 
respectively, and $g_1 = c_1$ (cf.\ \eqref{ch2.9}), we can apply
\eqref{5.21E} to $Q_{n-1}= G_{n+1} - z^{n+1} - c_1 F_n$ and hence obtain 
via \eqref{5.21E},  
\begin{equation}
\f{G_{n+1}(z,x)}{F_n(z,x)}= \f{z^{n+1}}{F_n(z,x)} + c_1+\sum_{j=1}^n
\f{G_{n+1}(\mu_j(x),x) - \mu_j(x)^{n+1}}{F_{n,z}(\mu_j(x),x)}\f{1}{z-\mu_j(x)}. \lb{5.21F}
\end{equation}
Next we recall some more Lagrange-type interpolation formulas: if 
$\{\mu_j\}_{j=1}^n \subset \bbC$, $\mu_j \neq \mu_k$ for $j \neq k$, 
$j,k = 1,\dots,n$, and $F_n(z) = \prod_{j=1}^n (z -\mu_j)$, then 
\begin{align}
& F_{n,z}(z) = \sum_{j=1}^n \prod_{\substack{\ell=1 \\ \ell \neq j}}^n (z - \mu_{\ell}), 
\quad F_{n,z} (\mu_j) = \prod_{\substack{\ell=1 \\ \ell \neq j}}^n (\mu_j - \mu_{\ell}),     
\lb{5.21Fa} \\
& \sum_{j=1}^n \f{\mu_j^{k-1}}{F_{n,z}(\mu_j)} = \delta_{k,n}, \;  
k =1,\dots,n, \quad 
\sum_{j=1}^n \f{\mu_j^{n}}{F_{n,z}(\mu_j)} = \sum_{j=1}^n \mu_j    \lb{5.21Fb} 
\end{align}
(see, e.g., \cite{GH03}, \cite[App.~E]{GH03a}). Then
\begin{align}
& \f{z^{n+1}}{F_n(z)} - \sum_{j=1}^n \f{\mu_j^{n+1}}{F_{n,z}(\mu_j)}\f{1}{z - \mu_j} 
= \sum_{j=1}^n \f{\big[z^{n+1} - \mu_j^{n+1}\big]}{F_{n,z}(\mu_j)}\f{1}{z - \mu_j}   \no \\
& \quad = \sum_{j=1}^n \f{\big[z^n + z^{n-1} \mu_j + z^{n-2} \mu_j^2 + \cdots 
+ z^2 \mu_j^{n-2} + z \mu_j^{n-1} + \mu_j^n\big]}{F_{n,z}(\mu_j)}   \no \\
& \quad = z + \sum_{j=1}^n \mu_j,
\end{align}
applying $a^{n+1} - b^{n+1} = (a-b)\big[a^n + a^{n-1} b + \cdots + a b^{n-1} + b^n\big]$ 
and \eqref{5.21Fb}. Thus,
\begin{equation}
\f{G_{n+1}(z,x)}{F_n(z,x)}= z + \sum_{j=1}^n \mu_j(x) + c_1+\sum_{j=1}^n
\f{G_{n+1}(\mu_j(x),x)}{F_{n,z}(\mu_j(x),x)}\f{1}{z-\mu_j(x)}.     \lb{5.21FA}
\end{equation}
Equivalently, employing \eqref{ch3.6}, \eqref{ch2.39D} for $c_1$, the trace formula 
\eqref{ch3.36} for $u$, and \eqref{5.21Fa}, 
\begin{equation}
G_{n+1}(z) - (z+2u) F_n(z) = - \sum_{j=1}^n 
\frac{y(\hat \mu_j)}{\prod_{\substack{\ell =1 \\\ell\ne j}}(\mu_j-\mu_\ell)} 
\prod_{\substack{k=1 \\k\neq j}}^n (z-\mu_k).     \lb{5.53a}
\end{equation} 
Alternatively, \eqref{5.53a} can be proved directly as follows: By the asymptotic behavior 
(cf.\ \eqref{ch3.55} for a refinement),
\begin{equation}
G_{n+1}(z) - (z+2u) F_n(z) \underset{|z| \to \infty}{=} \Oh\big(|z|^{n-1}\big)
\end{equation}
both sides of \eqref{5.53a} are polynomials in $z$ of degree $n-1$ which coincide (with the 
value $- y(\hat \mu_j)$ applying \eqref{ch3.6} again) at the $n$ points $\mu_j$, $j=1,\dots,n$. 

Employing \eqref{5.21FA} then yields 
\begin{align}
\begin{split} 
&\f{\mp R_{2n+2}(z)^{1/2} + G_{n+1}(z,x)}{F_n(z,x)} = 
\f{\mp R_{2n+2}(z)^{1/2}}{F_n(z,x)} + z + \sum_{j=1}^n \mu_j(x) + c_1  \\
& \quad + \sum_{j=1}^n \f{G_{n+1}(\mu_j(x),x)\{\varepsilon_j(x)+[1-\varepsilon_j(x)]\}}
{F_{n,z}(\mu_j(x),x)} \f{1}{z-\mu_j(x)}, \lb{5.21G}
\end{split} 
\end{align}
and hence \eqref{5.20a} follows by inserting \eqref{5.21D} into
\eqref{5.21G} and changing $z$ into $-z$. Equations \eqref{5.21A} and 
\eqref{5.21B} are clear from the matrix analog of \eqref{HFb}.

The statement \eqref{5.21a} for the essential half-line spectra then
follows from the fact that the measure in the Nevanlinna--Herglotz representation
\eqref{5.20a} of $\wti m_\pm$ (as a function of $z$) is supported on the
set $\Sigma_n$ in \eqref{5.9}, with a strictly positive density on the open
interior $\Sigma_n^o$ of $\Sigma_n$. The transformation $z \to - z^{-1}$ then
yields \eqref{5.21a} and since half-line spectra with a regular endpoint
$x$ have always simple spectra this completes the proof of our half-line
spectral assertions. The full-line case follows in exactly the same manner
since the corresponding $2\times 2$ matrix-valued measure $\Omega$ in the
Nevanlinna--Herglotz representation \eqref{5.21A} of $\wti M$ (as a function of
$z$) also has support $\Sigma_n$ and rank equal to two on $\Sigma_n^o$.
\end{proof}

Returning to direct spectral theory, we note that the two spectral  
problems \eqref{5.6} on $\bbR$ associated with the vanishing of the first
and second component of $\wti \Psi$ at $x$, respectively, are clearly
self-adjoint since they correspond to the choices $\beta=(1\;0)$ and
$\beta=(0\;1)$ in \eqref{4.5}. Hence a comparison with
\eqref{ch3.5},
\eqref{ch3.25}, and \eqref{ch3.26} necessarily yields 
$\{\mu_j(x)\}_{j=1,\dots,n}, \{\nu_j(x)\}_{j=1,\dots,n}\subset\bbR$. Thus
we will assume the  convenient eigenvalue orderings 
\begin{equation}
\mu_j(x) < \mu_{j+1}(x), \quad \nu_j(x) < \nu_{j+1}(x) \text{  for
$j=1,\dots,n-1$}, \; x\in\bbR.  \lb{5.34} 
\end{equation} 
 
Combining Lemma \ref{l5.2} with the Nevanlinna--Herglotz property of the 
$2\times 2$ Weyl--Titchmarsh matrix $\wti M(\dott,x)$ then yields the
following refinement of Lemma \ref{l3.2}.

\begin{theorem} \lb{t5.4}
Assume Hypothesis 5.1. Then $\{\hat\mu_j\}_{j=1,\dots,n}$, with
the projections $\mu_j(x)$, $j=1,\dots,n$, the zeros of $F_n(\dott,x)$ in
\eqref{ch3.5}, satisfies the first-order system of differential equations
\eqref{ch3.31} on
$\Omega_\mu=\bbR$ and 
\begin{equation}
\hat\mu_j\in C^\infty(\bbR,\calK_n),\quad j=1, \dots, n. \lb{5.36}
\end{equation}
Moreover,
\begin{equation}
\mu_j(x)\in[E_{2j-1},E_{2j}], \quad j=1,\dots,n, \; x\in\bbR.
\lb{5.37}
\end{equation}
In particular, $\hat \mu_j (x)$ changes sheets whenever it hits 
$E_{2j-1}$ or $E_{2j}$ and its projection $\mu_j (x)$ remains trapped in
$[E_{2j-1}, E_{2j}]$ for all $j=1,\dots,n$ and $x\in\bbR$. The analogous
statements apply to $\hat \nu_j(x)$ and one infers
\begin{equation}
\nu_j(x)\in[E_{2j-1},E_{2j}], \quad j=1,\dots,n, \; x\in\bbR. \lb{5.38}
\end{equation}
\end{theorem}
\begin{proof}
Since $\wti M(\dott,x)$ is a $2\times 2$ Nevanlinna--Herglotz matrix, its diagonal
elements are Nevanlinna--Herglotz functions. Thus, 
\begin{equation}
\wti M_{1,1}(\ti z,x) =\f{- F_n(-z,x)}{2R_{2n+2}(-z)^{1/2}}, \quad 
\wti M_{2,2}(\ti z,x) =\f{H_n(-z,x)}{2R_{2n+2}(-z)^{1/2}}, \quad \ti z= - z^{-1},    \lb{5.39}
\end{equation}
are Nevanlinna--Herglotz functions (the left-hand sides with respect to $\ti z$,
the right-hand sides with respect to $z$) and the interlacing properties
\eqref{5.37}, \eqref{5.38} then follow from \eqref{5.22} and \eqref{5.26}. 
\end{proof}

\begin{remark} \lb{r5.5a}
The Nevanlinna--Herglotz property of $\wti M_{2,2}(\, \cdot \,,x)$ necessitates the inequality 
$h_0(x) < 0$, $x \in \bbR$, which appears to be difficult to verify directly. In particular, the explicit 
expression (cf.\ \eqref{ch3.54} with $H=H_n$, and \eqref{ch3.55})
\begin{align}
\begin{split} 
h_0 &= \sum _{j=0}^{2n+1} \sum _{k=j+1}^{2n+1} E_j E_k - \left(\f12 \sum_{j=0}^{2n+1} E_j\right)^2 \\
& \quad -\Bigg[ \sum_{\substack{j_1,j_2=1\\ j_1 < j_2}}^n [\mu_{j_1} \mu_{j_2,x} 
+ \mu_{j_1,x} \mu_{j_2} + 2 \mu_{j_1}\mu_{j_2}] - 2\alpha \sum_{j=1}^n \mu_j\Bigg]
\end{split} 
\end{align}
does not necessarily shed any light on this issue. 
\end{remark}

\begin{remark} \lb{r5.6}
The zeros $\mu_j(x)\in (E_{2j-1},E_{2j})$, $j=1,\dots,n$ of
$F_n(\dott,x)$ which are related to eigenvalues of the Hamiltonian system
\eqref{5.6} on $\bbR$ associated with the boundary condition
$\wti\psi_1(x)=0$, in fact, are related to left and right half-line
eigenvalues of the corresponding Hamiltonian system restricted to the
half-lines $(-\infty,x]$ and $[x,\infty)$, respectively. Indeed, by
\eqref{5.17} and \eqref{5.19a}, depending on whether
$\hat\mu_j(x)\in\Pi_{n,+}$ or $\hat\mu_j(x)\in\Pi_{n,-}$, $\mu_j(x)$ is related to
a left or right half-line eigenvalue associated with the Dirichlet
boundary condition $\wti\psi_1(x)=0$. A careful investigation of the sign
of the right-hand sides of the Dubrovin equations \eqref{ch3.30}
$($combining \eqref{5.1}, \eqref{5.11}, and \eqref{5.13}$)$, then proves
that the $\mu_j(x)$ related to right $($resp., left\,$)$ half-line
eigenvalues of the Hamiltonian system \eqref{5.6} associated with the
boundary condition $\wti\psi_1(x)=0$, are strictly monotone increasing
$($resp., decreasing\,$)$ with respect to $x$, as long as the $\mu_j$ stay
away from the right $($resp., left\,$)$ endpoints of the corresponding 
spectral gaps $(E_{2j-1},E_{2j})$. Here we purposely avoided the limiting
case where some of the $\mu_k(x)$ hit the boundary of the spectral gaps, 
$\mu_k(x)\in\{E_{2k-1},E_{2k}\}$, since the half-line eigenvalue
interpretation is lost as there is no $L^2((x,\pm\infty))^2$
eigenfunction $\wti\Psi(x)$ satisfying $\wti\psi_1(x)=0$ in this case. In
fact, whenever an eigenvalue $\mu_k(x)$ hits a spectral gap endpoint, the
associated point $\hat\mu_j(x)$ on $\calK_n$ crosses over from one
sheet to the other $($equivalently, the corresponding left half-line
eigenvalue becomes a right half-line eigenvalue and vice versa\,$)$ and
accordingly, strictly increasing half-line eigenvalues become strictly
decreasing half-line eigenvalues and vice versa. In particular, using the
appropriate local coordinate $(z-E_{2k})^{1/2}$ $($resp.,
$(z-E_{2k-1})^{1/2}$$)$ near
$E_{2k}$ $($resp., $E_{2k-1}$$)$, one verifies that $\mu_k(x)$ does not
pause at the endpoints $E_{2k}$ and $E_{2k-1}$. \hfill $\diamond$
\end{remark}

Next, we turn to the inverse spectral problem and determine the
isospectral manifold of real-valued, smooth, and bounded $\CH$ solutions.

Our basic assumptions then will be the following:

\begin{hypothesis} \lb{h5.7}
Suppose 
\begin{equation}
E_0 < E_1 < \cdots < E_{2n} < E_{2n+1}, \quad 0 \in (E_{2m_0}, E_{2m_0 + 1})     \lb{5.42}
\end{equation}
for some $m_0 \in \{0,\dots,n\}$. In addition, fix $x_0\in\bbR$, and assume 
that the initial data
\begin{equation}
\{\hat\mu_j(x_0)=(\mu_j(x_0),- G_{n+1}(\mu_j(x_0),x_0))\}_{j=1,\dots,n} \subset\calK_n \lb{5.43}
\end{equation}
for the Dubrovin equations \eqref{ch3.31} are constrained by
\begin{equation}
\mu_j(x_0)\in[E_{2j-1},E_{2j}], \quad j=1,\dots,n. \lb{5.44}
\end{equation}
\end{hypothesis}

\begin{theorem} \lb{t5.8}
Assume Hypothesis \ref{h5.7}. Then the Dubrovin initial value problem
\eqref{ch3.31}, \eqref{5.43}, \eqref{5.44} has a unique solution
$\{\hat\mu_j\}_{j=1,\dots,n}\subset\calK_n$ satisfying
\begin{equation}
\hat\mu_j\in C^\infty(\bbR,\calK_n),\quad j=1, \dots, n,
\lb{5.45}
\end{equation}
and the projections $\mu_j$ remain trapped in the intervals
$[E_{2j-1},E_{2j}]$, $j=1,\dots,n$, for all $x\in\bbR$,
\begin{equation}
\mu_j(x)\in[E_{2j-1},E_{2j}], \quad j=1,\dots,n, \; x\in\bbR. \lb{5.46}
\end{equation}
Moreover, $u, \kap$ defined by the formulas \eqref{ch3.36}, \eqref{ch3.36a}, 
\begin{align}
u(x)&=\f12\sum_{j=1}^n\mu_{j}(x)-\f14\sum_{m=0}^{2n+1} E_{m}, \quad
\kap(x) = - \Bigg(\prod_{m=0}^{2n+1}
E_m\Bigg)\Bigg( \prod_{j=1}^n \mu_j(x)^{-2}\Bigg), \quad x\in\bbR,   \lb{5.47}
\end{align}
satisfy Hypothesis \ref{h5.1}, that is, $(u, \kap)$ is a real-valued solution  
of the $n$th stationary $\CH$ equation \eqref{ch2.29a},
\begin{equation}
\sCH_{n}(u, \kap)=0   \lb{5.51}
\end{equation} 
with integration constants $c_\ell$ in \eqref{5.51} given by
$c_\ell=c_\ell(\ul E)$, $\ell=1,\dots,n$, according to \eqref{ch2.39C},
\eqref{ch2.39D}, satisfying 
\begin{equation}
u, \kap \in C^\infty(\bbR), \; \kap > 0, \; u^{(m)}, \kap^{(m)} \in L^\infty(\bbR), \; 
m\in\bbN_{0}.    \label{5.48} 
\end{equation}
\end{theorem}
\begin{proof}
Given initial data constrained by $\mu_j(x_0)\in(E_{2j-1},E_{2j})$, 
$j=1,\dots,n$, one concludes from the Dubrovin equations \eqref{ch3.31} and
the sign properties of $R_{2n+2}^{1/2}$ on the intervals
$[E_{2k-1},E_{2k}]$, $k=1,\dots,n$, described in \eqref{5.11}, that the
solution $\mu_j(x)$ remains in the interval $[E_{2j-1},E_{2j}]$ as long as
$\hat\mu_j(x)$ stays away from the branch points $(E_{2j-1},0),
(E_{2j},0)$. In case $\hat \mu_j$ hits  such a branch point, one can use
the local chart around $(E_m,0)$, with local coordinate $\zeta=\sigma
(z-E_m)^{1/2}$, $\sigma\in\{1,-1\}$, $m\in\{2j-1,2j\}$, to verify
\eqref{5.45} and \eqref{5.46}. Relations \eqref{5.47}, \eqref{5.48} are
then evident from \eqref{5.45}, \eqref{5.46}, and
\begin{equation}
|\partial_x^k \mu_j(x)|\leq C_k, \quad k\in\bbN_0, \; j=1,\dots,n, \;
x\in\bbR. \lb{5.52}
\end{equation}
In the course of the proof of Theorem \ref{t3.10} one
constructs the polynomial formalism ($F_n$, $G_{n+1}$, $H_n$, $R_{2n+2}$, etc.)
and then obtains identity \eqref{ch3.36a} as an elementary consequence. Finally, \eqref{5.51} also follows from Theorem \ref{t3.10} (with $\Omega_\mu=\bbR$).
\end{proof}

\begin{corollary} \lb{c5.9}
Fix $\{E_m\}_{m=0,\dots, 2n+1}\subset\bbR$ and assume the ordering 
\begin{equation}
E_0 < E_1 < \cdots < E_{2n} < E_{2n+1}, \quad 0 \in (E_{2m_0}, E_{2m_0 + 1})     \lb{5.42A}
\end{equation}
for some $m_0 \in \{0,\dots,2n+1\}$. Then the isospectral manifold of smooth,  
real-valued solutions $u, \kap \in C^\infty(\bbR)$, $\kap > 0$, of
$\sCH_{n}(u, \kap)=0$ is given by the real $n$-dimensional torus $\bbT^n$. $($These 
smooth solutions necessarily satisfy 
$u^{(m)}, \kap^{(m)}\in L^\infty(\bbR)$, $m\in\bbN_{0}$.$)$
\end{corollary}
\begin{proof}
The discussion in Remark \ref{r5.6} and Theorem \ref{t5.8}, shows that the
motion of each $\hat\mu_j(x)$ on $\calK_n$ proceeds topologically on a
circle and is uniquely determined by the initial data $\hat\mu_k(x_0)$,
$k=1,\dots,n$. More precisely, the initial data 
\begin{align}
\begin{split}
&\hat\mu_j(x_0)=(\mu_j(x_0),y(\hat\mu_j(x_0)))=(\mu_j(x_0),
-G_{n+1}(\mu_j(x_0),x_0)), \lb{5.53} \\
&\mu_j(x_0)\in[E_{2j-1},E_{2j}], \quad j=1,\dots,n,  
\end{split}
\end{align}
are topologically equivalent to data of the type 
\begin{align}
(\mu_j(x_0),\sigma_j(x_0))\in [E_{2j-1},E_{2j}]\times \{+,-\},  
\quad j=1,\dots,n, \lb{5.54}
\end{align}
the sign of $\sigma_j(x_0)$ depending on $\hat\mu_j(x_0)\in\Pi_{n,\pm}$. If, 
on the other hand, 
some of the $\mu_k(x_0)\in\{E_{2k-1},E_{2k}\}$, then the determination of
the sheet $\Pi_{n,\pm}$ and hence the sign $\sigma_k(x_0)$ in \eqref{5.54}
becomes superfluous and is eliminated from \eqref{5.54}. Indeed, since by
\eqref{ch2.23}, 
\begin{equation}
G_{n+1}(\mu_j(x_0),x_0)^2 = R_{2n+2}(\mu_j(x_0)), \lb{5.52a}
\end{equation}
$G_{n+1}(\mu_j(x_0),x_0)$ is determined up to a sign unless $\mu_j(x_0)$ hits
a spectral gap endpoint $E_{2j-1}, E_{2j}$ in which case
$G_{n+1}(\mu_j(x_0),x_0)=R_{2n+2}(\mu_j(x_0))=0$ and the sign ambiguity
disappears. The $n$ data in \eqref{5.54} (properly interpreted if
$\mu_j(x_0)\in\{E_{2j-1},E_{2j}\}$) can be identified with circles. Since
the latter are  independent of each other, the isospectral manifold of
real-valued, smooth, and bounded solutions of
$\sCH_{n}(u)=0$ is given by the real $n$-dimensional torus $\bbT^n$.
\end{proof}

In summary, one observes that the reality problem for smooth bounded
solutions of the $\CH$ hierarchy, assuming the ordering \eqref{5.42}, parallels 
that of the KdV hierarchy with the basic self-adjoint Lax operator (the
one-dimensional Schr\"odinger operator) replaced by the self-adjoint
Hamiltonian system \eqref{5.6}. 

\begin{remark} \lb{r5.10}
Since the focus of this paper centered around the two-component Camassa--Holm hierarchy $\CH$, we assumed $\kap > 0$ throughout Section \ref{s5}. The limit 
$\kap \to 0$, although straightforward in connection with the material in 
Section \ref{s2}, requires some care in Sections \ref{s3} and \ref{s5}. Indeed, 
formula \eqref{5.47} for $\kap$ indicates the singular nature of such a limit: one 
infers from the $z^1$-term in \eqref{ch2.23}, \\
\begin{equation}
f_n \big[2 \alpha g_n + h_{n-1} - (\alpha^2+\kap) f_{n-1}\big] = - \sum_{m=0}^{2n+1} 
\prod_{\substack{m'=0 \\ m' \neq m}}^{2n+1} E_{m'}, 
\end{equation} 
and recalling $f_n = (-1)^n \prod_{j=1}^n \mu_j$, one concludes 
\begin{equation}
\big[2 \alpha g_n + h_{n-1} - (\alpha^2+\kap) f_{n-1}\big] = (-1)^{n+1} 
\Bigg(\sum_{m=0}^{2n+1} \prod_{\substack{m'=0 \\ m' \neq m}}^{2n+1} E_{m'}\Bigg)
\prod_{j=1}^n \mu_j^{-1}.    \lb{5.75} 
\end{equation}

The case $\kap = 0$ has been analyzed in detail in \cite{GH03}, \cite[Ch.~5]{GH03a} 
and under the assumption $[4u - u_{xx}] > 0$, it is known that necessarily 
\begin{align}
E_0 < E_1 < \cdots < E_{2n_0} < E_{2n_0+1} = 0, \quad 
\mu_{j_0}(x) \in [E_{2j_0-1}, E_{2j_0}], \; j_0 = 1,\dots,n_0,   \lb{5.76}
\end{align}
with $n_0$ the (topological) genus of the underlying curve $\calK_{n_0}$. (If $[4u - u_{xx}] < 0$  
all $E_m$ are reflected with respect to $E=0$.) A comparison with the case when $\kap > 0$ and $[4u - u_{xx}] > 0$, 
\begin{align}
\begin{split}
& E_0 < E_1 < \cdots < E_{2n} < E_{2n+1}, \quad 0 \in (E_{2m_0}, E_{2m_0 + 1}),        \lb{5.77}   \\
& \mu_j(x) \in [E_{2j-1}, E_{2j}], \; j=1,\dots,n
\end{split} 
\end{align}
for some $m_0 \in \{0,\dots,n\}$, shows that in order to guarantee a smooth limit 
$\kap \to 0$ in \eqref{5.75}, the following situation must occur,
\begin{align}
E_{2m_0 + 2k -1}, \mu_{m_0 + k}, E_{2m_0 + 2k}, E_{2n+1} \downarrow 0 
\, \text{ as $\kap \to 0$}, \, k = 1,\dots, n - m_0. 
\end{align} 
The limit $\kap \to 0$ thus yields a singular curve $\calK_n$ associated with 
\begin{equation}
y^2 = R_{2n+2}(z) = \bigg(\prod_{m=0}^{2m_0 +1} (z - E_m)\bigg) z^{2(n-m_0)}.
\end{equation}
Desingularizing this curve yields $y^2 = \prod_{m=0}^{2m_0 +1} (z - E_m)$ and hence 
corresponds to $m_0 = n_0$ in connection with \eqref{5.76}. \hfill $\diamond$
\end{remark}

\medskip

Finally, we briefly turn to the time-dependent case. 

\begin{hypothesis}\label{h5.11} 
Suppose that $u, \kap \colon \bbR^2\to\bbC$ satisfy 
\begin{align}
&u(\dott,t), \kap(\dott,t) \in C^\infty(\bbR), \; 
\f{\partial^m u}{\partial x^m}(\dott,t), \f{\partial^m \kap}{\partial x^m}(\dott,t) 
\in L^\infty(\bbR), \; m\in\bbN_{0}, \, t\in\bbR, \no \\
&u(x,\dott), u_{x}(x,\dott), \kap(x,\dott) \in C^1(\bbR), \; x\in\bbR. \lb{ch5.100}
\end{align}
\end{hypothesis}

The basic problem in the analysis of algebro-geometric solutions of 
the $\CH$ hierarchy consists in solving the time-dependent $r$th $\CH$
flow with initial data a stationary solution of the $n$th equation
in the hierarchy. More precisely, given $n\in\bbN_0$, consider a solution
$\big(u^{(0)}, \kap^{(0)}\big)$ of the $n$th stationary $\CH$ equation, that is,  
$\sCH_n\big(u^{(0)}, \kap^{(0)}\big)=0$ 
associated with $\calK_n$ and a given set of integration constants
$\{c_\ell\}_{\ell=1,\dots,n}\subset\bbC$. Next, let $r\in\bbN_0$; we
intend to construct a solution $u$, $\kap$ of the $r$th $\CH$ flow $\CH_r(u,\kap)=0$
with $u(t_{0,r})=u^{(0)}$, $\kap(t_{0,r})=\kap^{(0)}$ for some $t_{0,r}\in\bbR$. 
To emphasize that the
integration constants in the definitions of the stationary and the
time-dependent $\CH$ equations are independent of each other, we
indicate this by adding a tilde on all the time-dependent quantities.
Hence we shall employ the notation $\wti V_r$, $\wti F_r$,
$\wti G_{r}$, $\wti H_r$, $\ti f_{s}$, $\ti g_{s}$, $\ti h_{s}$,
 $\ti c_{s}$, etc., in order to distinguish them from 
$V_n$, $F_n$, $G_{n+1}$, $H_n$, $f_{\ell}$, $g_{\ell}$, $h_{\ell}$, 
$c_{\ell}$, etc., in the following. In addition,  we will follow a more
elaborate notation inspired by Hirota's $\tau$-function approach and
indicate the individual $r$th $\CH$ flow by a separate time variable 
$t_r \in \bbR$. 
 
Summing up, we are seeking a solution $u$ of
\begin{align}
& \wti \CH_{r}(u,\kap)= \begin{pmatrix} u_{t_r} + \frac{1}{2} \ti f_{r+1,x} \\[1mm] 
\kap_{t_r}  - \kap_x \ti f_r - 2 \kap \ti f_{r,x} \end{pmatrix} = 0,     \lb{ch5.101} \\ 
&u(x,t_{0,r})=u^{(0)}(x), \; \kap(x,t_{0,r})=\kap^{(0)}(x), \quad x\in\bbR, \no \\
& \sCH_n\big(u^{(0)}, \kap^{(0)}\big) = \begin{pmatrix} 
\frac{1}{2} f_{n+1,x}^{(0)} \\[1mm] 
- \kap_x^{(0)} f_n^{(0)} - 2 \kap^{(0)} f_{n,x}^{(0)}
\end{pmatrix} = 0,    \lb{ch5.102}
\end{align}
for some $t_{0,r}\in\bbR$, $n,r\in\bbN_0$, where $(u, \kap)$
satisfy \eqref{ch5.100}. 

We pause for a moment to reflect on the pair of equations \eqref{ch5.101},
\eqref{ch5.102}: As it turns out (cf.\ \cite{GH03}, \cite[Sect.~5.4]{GH03a} in the special case $\kap=0$), it represents a dynamical system on the set
of algebro-geometric solutions isospectral to the initial value 
$\big(u^{(0)}, \kap^{(0)}\big)$. By
isospectrality we here allude to the fact that for any fixed $t_r\in\bbR$,
the solution $(u(\dott,t_r), \kap(\dott,t_r))$ of \eqref{ch5.101}, \eqref{ch5.102} is a
stationary solution of \eqref{ch5.102}, 
\begin{align}
\sCH_n(u(\dott,t_r), \kap(\dott,t_r)) &= \begin{pmatrix} \frac{1}{2} f_{n+1,x}^{(0)}(\dott,t_r) \\[1mm] 
- \kap_x^{(0)}(\dott,t_r) f_n^{(0)}(\dott,t_r) - 2 \kap^{(0)}(\dott,t_r) f_{n,x}^{(0)}(\dott,t_r) 
\end{pmatrix}    \no \\
&= 0,    \lb{ch5.102A} 
\end{align}
associated with the fixed underlying algebraic curve $\calK_n$. Put
differently, the solution $(u(\dott,t_r), \kap(\dott,t_r))$ is an isospectral deformation of 
$\big(u^{(0)}, \kap^{(0)}\big)$ with $t_r$ the corresponding deformation parameter. 
In particular, $(u(\dott,t_r), \kap(\dott,t_r))$ traces out a curve in the set of 
algebro-geometric solutions isospectral to $\big(u^{(0)}, \kap^{(0)}\big)$.

Thus, relying on this isospectral property
of the $\CH$ flows, we will go a step further and assume
 \eqref{ch5.102} not only at $t_r=t_{0,r}$ but for all $t_r\in\bbR$.
Hence, we start with
\begin{align}
U_{t_r}(z,x,t_r)-\wti V_{r,x}(z,x,t_r)
+[U(z,x,t_r),\wti V_r(z,x,t_r)]&=0, \lb{ch5.103} \\ 
-V_{n,x}(z,x,t_r)+[U(z,x,t_r),V_n(z,x,t_r)]&=0, \lb{ch5.103A} \\
& \hspace*{-2.15cm} (z,x,t_r)\in\bbC\times\bbR^2, \no 
\end{align}
where (cf. \eqref{ch2.26})
\begin{align}
U(z,x,t_r)&= - z^{-1} \begin{pmatrix} \alpha(x,t_r) & -1\\
\alpha(x,t_r)^2 + \kap(x,t_r) &- \alpha(x,t_r) \end{pmatrix} 
+ \begin{pmatrix} -1 & 0 \\
0 & 1 \end{pmatrix}, \no\\
\wti V_r(z,x,t_r)&= z^{-1} \begin{pmatrix} -\wti G_{r+1}(z,x,t_r)& \wti
F_{r}(z,x,t_r) \\ \wti H_{r}(z,x,t_r) & \wti G_{r+1}(z,x,t_r)
\end{pmatrix}, \lb{ch5.103B} \\
V_n(z,x,t_r)&= z^{-1} \begin{pmatrix} -G_{n+1}(z,x,t_r)& F_{n}(z,x,t_r)\\
 H_{n}(z,x,t_r) &G_{n+1}(z,x,t_r) \end{pmatrix}, \no 
\end{align}
and
\begin{align}
F_n(z,x,t_r)&=\sum_{\ell=0}^n f_{n-\ell}(x,t_r)z^\ell=
\prod_{j=1}^n (z-\mu_j(x,t_r)), \lb{ch5.106f} \\
G_{n+1}(z,x,t_r)&=\sum_{\ell=0}^{n+1} g_{n+1-\ell}(x,t_r)z^\ell - f_{n+1}(x,t_r) - \f12 f_{n+1,x}(x,t_r),  \lb{ch5.106g} \\
H_n(z,x,t_r)&=\sum_{\ell=0}^n h_{n-\ell}(x,t_r)z^\ell + g_{n+2,x}(x,t_r) \no\\ & =
h_0(x,t_r)\prod_{j=1}^n (z-\nu_j(x,t_r)), \lb{ch5.106h} \\
\wti F_r(z,x,t_r)&=\sum_{s=0}^r \ti f_{r-s}(x,t_r)z^s,
\lb{ch5.106j} \\
\wti G_{r+1}(z,x,t_r)&=\sum_{s=0}^{r+1} \ti g_{r+1-s}(x,t_r)z^s - \ti f_{r+1}(x,t_r) - \f12 \ti f_{r+1,x}(x,t_r), 
\lb{ch5.106k} \\ 
\wti H_r(z,x,t_r)&=\sum_{s=0}^r \ti h_{r-s}(x,t_r)z^s + \ti g_{r+2,x}(x,t_r) ,
\lb{ch5.106l}
\end{align}
for fixed $n,r\in\bbN_0$. Here $f_\ell(x,t_r)$, $\ti f_s(x,t_r)$,
$g_\ell(x,t_r)$, $\ti g_s(x,t_r)$, $h_\ell(x,t_r)$, and $\ti h_s(x,t_r)$ 
for $\ell=0,\dots,n+1$, $s=0,\dots,r+1$, are defined as in
\eqref{ch2.3} and \eqref{ch2.8} with $u(x)$ replaced by
$u(x,t_r)$, etc., and with appropriate integration constants. Explicitly,
\eqref{ch5.103}, \eqref{ch5.103A} are equivalent to 
\begin{align}
z \wti F_{r,x}(z,x,t_r) &= - 2[\alpha(x,t_r) +z] \wti F_r(z,x,t_r) + 2 \wti G_{r+1}(z,x,t_r),  
\lb{ch5.106b} \\
z \wti G_{r+1,x}(z,x,t_r) &= z \alpha_{t_r}(x,t_r)  
- \big[\alpha(x,t_r)^2 + \kap(x,t_r)\big] \wti F_r(z,x,t_r) - \wti H_{r}(z,x,t_r),    \lb{ch5.106a} \\ 
z \wti H_{r,x}(z,x,t_{r})  
&= - z [2 \alpha(x,t_r) \alpha_{t_r}(x,t_r) + \kap_{t_r}(x,t_r)]     \no \\
&\quad + 2 [\alpha(x,t_r) + z] \wti H_{r}(z,x,t_{r})    \lb{ch5.106c} \\ 
& \quad + 2 \big[\alpha(x,t_r)^2 + \kap(x,t_r)\big] \wti G_{r+1}(z,x,t_{r})=0,    \no 
\intertext{and} 
z F_{n,x}(z,x,t_r)&= - 2 [\alpha(x,t_r) + z] F_n(z,x,t_r) + 2 G_{n+1}(z,x,t_r), 
\lb{ch5.104a} \\ 
z G_{n+1,x}(z,x,t_r)&= - \big[\alpha(x,t_r)^2 + \kap(x,t_r)\big] F_{n}(z,x,t_r) 
- H_n(z,x,t_r), \lb{ch5.104b} \\ 
z H_{n,x}(z,x,t_r)&=2 [\alpha(x,t_r) + z] H_{n}(z,x,t_r) \no \\
& \quad + 2\big[\alpha(x,t_r)^2 + \kap(x,t_r)\big] G_{n+1}(z,x,t_r).
\lb{ch5.104c} 
\end{align}

One observes that equations \eqref{ch2.3}--\eqref{ch2.39a} apply to 
$F_n$, $G_{n+1}$, $H_n$, $f_\ell$, $g_\ell$, and $h_\ell$ and  
\eqref{ch2.3}--\eqref{ch2.9}, \eqref{ch2.26}, with 
$n$ replaced by $r$ and $c_{\ell}$ replaced by $\ti c_{\ell}$, apply
to $\wti F_r$, $\wti G_{r+1}$, $\wti H_r$, $\ti f_\ell$, $\ti g_\ell$,
and $\ti h_\ell$. In particular, the fundamental identity \eqref{ch2.23},
\begin{equation}
G_{n+1}(z,x,t_r)^2 + F_n(z,x,t_r) H_n(z,x,t_r)=R_{2n+2}(z), 
\quad t_r\in\bbR,  \lb{ch5.105}
\end{equation}
holds as in the stationary context and the hyperelliptic curve $\calK_n$ is still given by
\begin{equation}
\calK_n \colon \calF_n(z,y)=y^2-R_{2n+2}(z)=0, \quad 
R_{2n+2}(z)=\prod_{m=0}^{2n+1} (z-E_m), \quad 
\{E_m\}_{m = 0}^{2n+1} \subset \bbC. \lb{ch5.107}
\end{equation}
Moreover, \eqref{ch5.103} and \eqref{ch5.103A} also yield 
\begin{equation}
-V_{n,t_r}(z,x,t_r) + \big[\wti V_r(z,x,t_r),V_n(z,x,t_r)\big] = 0, \quad 
(z,x,t_r)\in\bbC\times\bbR^2.  
\end{equation}

The independence of \eqref{ch5.105} of
$t_r\in\bbR$ can be interpreted as follows: The $r$th $\CH$ flow represents an
isospectral deformation of the curve $\calK_n$ in \eqref{ch5.107}, in particular, the branch points of $\calK_n$ remain invariant under these flows, 
\begin{equation}
\partial_{t_r} E_m =0, \quad m=0,\dots,2n+1. \lb{ch5.109}
\end{equation}
Without going into further details we note that the time-dependent analog of the Dubrovin-type equations \eqref{ch3.31} now reads,
\begin{align}
\mu_{j,x}(x,t_r)&=2\mu_j(x,t_r)^{-1}y(\hat\mu_j(x,t_r))
\prod_{\substack{\ell=1\\ \ell\neq j}}^n [\mu_j(x,t_r)-\mu_\ell(x,t_r)]^{-1},
\lb{ch5.110}\\
\mu_{j,t_r}(x,t_r)&=2\ti F_r(\mu_j(x,t_r),x,t_r) \no \\
& \quad \times \mu_j(x,t_r)^{-1} y(\hat\mu_j(x,t_r))
\prod_{\substack{\ell=1\\ \ell\neq
j}}^n [\mu_j(x,t_r)-\mu_\ell(x,t_r)]^{-1}, \lb{ch5.111} \\  
& \hspace*{4.5cm} j=1, \dots, n, \, (x,t_r)\in\wti\Omega_\mu, \no
\end{align}
with an appropriate open and connected set $\wti\Omega_\mu\subseteq\bbR^2$.  
In particular, higher-order $\CH_r$ flows drive each $\hat\mu_j(x,t_r)$ around the 
same circles as in the stationary case.

Together with the comments following \eqref{ch5.102}, this shows that isospectral torus questions are conveniently reduced to the study of the stationary hierarchy of $\CH$ flows since time-dependent solutions just trace out a curve in the isospectral torus defined by the stationary hierarchy. This is of course in complete agreement with other completely integrable $1+1$-dimensional hierarchies such as the KdV, Toda, and AKNS hierarchies.

\medskip

\noindent 
{\bf Acknowledgments.}
We are indebted to Christer Bennewitz for kindly sharing unpublished material on the $\CH$ system 
\eqref{ch1.0} with us, and to Don Hinton for helpful discussions on Hamiltonian systems. 

F.G.\ gratefully acknowledges kind invitations to the Faculty of Mathematics, University of Vienna, Austria, for parts of June 2014, and to the Department of Mathematical Sciences of the Norwegian University of Science and Technology, Trondheim, for parts of May and June 2015. The extraordinary hospitality by Gerald Teschl and Helge Holden at each Institution, as well as the stimulating atmosphere at both places, are greatly appreciated. 

\medskip


\end{document}